\documentclass[12pt]{article}
\usepackage{color}

\usepackage{amsfonts,amsmath,amsxtra,amsthm,amssymb,latexsym}
\usepackage{amssymb}  
\usepackage{bbm} 
\usepackage{graphicx}

 \usepackage[utf8]{inputenc}

\newtheorem{thm}{Theorem}[section]
 \newtheorem{cor}[thm]{Corollary}
 \newtheorem{lem}[thm]{Lemma}
 \newtheorem{prop}[thm]{Proposition} 

 \theoremstyle{definition}
 \newtheorem{defn}{Definition}[section]
 \theoremstyle{remark}
 \newtheorem{rem}{Remark}[section]
 
 \newtheorem{ex}[thm]{Example}

 \numberwithin{equation}{section}

\let\emptyset\varnothing

\DeclareMathOperator{\im}{Im}
\DeclareMathOperator{\re}{Re}

\DeclareMathOperator{\Ln}{Ln}
\DeclareMathOperator{\diam}{diam}
\DeclareMathOperator{\Arg}{Arg}
\DeclareMathOperator{\supp}{supp}
\DeclareMathOperator{\I}{Is}

\def\RR{\mathbb R}
\def\CC{\mathbb C}
\def\NN{\mathbb N}
\def\ZZ{\mathbb Z}
\def\QQ{\mathbb Q}

\def\DD{\mathbb D}

\def\BB{\mathbb B}
\def\AAA{\mathbb{F}}
\def\Aa{\mathbb{A}}

\def\ii{\mathrm{i}}
\def\dd{\mathrm{d}}
\def\ee{\mathrm{e}}

\def\wt{\widetilde}

\def\al{\alpha}

\def\vphi{\varphi}
\def\la{\lambda}
\def\th{\theta}
\def\de{\delta}
\def\ep{\epsilon}
\def\vep{\varepsilon}
\def\ga{\gamma}
\def\Ga{\Gamma}
\def\pa{\partial}
\def\om{\omega}

\def\De{\Delta}
\def\Si{\Sigma}
\def\si{\sigma}

\def\L{\mathcal{L}}
\def\E{\mathcal{E}}

\def\N{\mathfrak{N}}

\def\V{\mathcal{V}}

\def\Z{\mathcal{Z}}
\def\K{\mathcal{K}}
\def\G{\mathcal{G}}

\def\Size{\mathfrak{s}}

\def\d{\mathrm{Ad}}
\def\gen{\mathrm{gen}}
\def\Sf{\mathfrak{S}}
\def\g{\wt r}
\def\M{\mathcal{M}}
\def\Gr{\mathfrak{G}}

\def\loc{\mathrm{loc}}

\def\hor{\mathrm{hor}}
\def\cont{\mathrm{cont}}
\def\comp{\mathrm{comp}}
\def\narrow{\mathrm{narrow}}

\def\id{\mathfrak{e}}
\def\Fr{\mathrm{Fr}}
\def\OO{\mathcal{O}}
\def\cub{\mathrm{cub}}

\begin{document}
\title{On the multilevel internal structure of the asymptotic distribution of resonances}
\author{}
\date{}
\maketitle

{\center 
{\large 
Sergio Albeverio $^{\text{a}}$ and Illya M. Karabash $^{\text{b,c,*}}$
\\[4ex]
}}

\vspace{4ex}

{\small \noindent
$^{\text{a}}$  Institute for Applied Mathematics, Rheinische Friedrich-Wilhelms Universität Bonn,
and Hausdorff Center for Mathematics, Endenicher Allee 60,
D-53115 Bonn, Germany\\[1mm]
$^{\text{b}}$
Mathematical Institute, Rheinische Friedrich-Wilhelms Universität Bonn,
 Endenicher Allee 60, D-53115 Bonn, Germany \\[1mm]
 $^{\text{c}}$ Institute of Applied Mathematics and Mechanics of NAS of Ukraine,
Dobrovolskogo st. 1, Slovyans'k 84100, Ukraine\\[1mm]
$^{\text{*}}$ Corresponding author: i.m.karabash@gmail.com\\[2mm]
E-mails: 
albeverio@iam.uni-bonn.de, i.m.karabash@gmail.com
}

\begin{abstract}
We prove that the asymptotic distribution of resonances has a multilevel internal structure for the following
classes of Hamiltonians $H$: Schrödinger operators with point interactions in $\RR^3$, quantum graphs, and 1-D photonic crystals.
In the case of $N \ge 2$ point interactions, the set of resonances $\Si (H)$ essentially consists of a finite number of sequences 
with logarithmic asymptotics. We show how the leading parameters $\mu$ of these sequences are connected with the geometry of 
the set $Y=\{y_j\}_{j=1}^N$ of interaction centers $y_j \in \RR^3$.  The minimal parameter $\mu^{\min}$ corresponds to 
the sequences with `more narrow' and so more observable resonances. 
The asymptotic density of such narrow resonances can be expressed via 
the multiplicity of $\mu^{\min}$, 
which occurs to be connected with the symmetries 
of $Y$ and naturally introduces a finite number of classes of configurations of $Y$.
In the case of quantum graphs and 1-D photonic crystals, 
the decomposition of $\Si(H)$ into a finite number of asymptotic sequences is proved under additional 
commensurability conditions. To address the case of a general quantum graph, we introduce families 
of special asymptotic density functions for two classes of strips in $\CC$.
The obtained results and effects are compared with those of 
obstacle scattering.
\end{abstract}

{
\small \noindent
MSC-classes: 
35B34, 
35P20, 
35J10, 
35P25, 
81Q37, 
81Q35, 
78A45,  
81Q80, 
70J10,  
05C90 
\\[2mm]

\tableofcontents

}

\normalsize

\section{Introduction}
\label{s:intro}

\subsection{Main goals and related studies}

Let $\De= \sum_{j=1}^N \pa^2_{x_j}$ be the Laplacian operator 
in the complex Lebesgue space $L^2_\CC (\RR^m)$ with odd $m \ge 1$.
For operators $H$ obtained as various types of perturbations of $(-\De)$ on compact subsets of $\RR^m$, \emph{resonances} $k$  
are defined as poles of the resolvent $(H-z^2)^{-1}$ extended in a generalized sense through $\RR$ into the 
 lower complex half-plane $\CC_- := \{ z \in \CC : \im z <0\}$ (this and other types of definitions can be found, e.g., in \cite{AH84,DZ17,Z99,Z17}).

Reviews on resonances for obstacle and geometrical scattering can be found in \cite{DZ17,Z99,Z17}.
Due to various  engineering applications, wave equations, resonances, and related optimization problems for  
noncompact quantum graphs (see \cite{KS99,EL10,DP11,DEL10,GSS13,LZ16,L16,EL17,dVT18} and the monographs \cite{BK13,P12})
and for point interactions (see \cite{AH84,DK07,Ka14,AK17,HL17,LL17,AK18} and the monographs \cite{AGHH12,AK00}) 
attracted a substantial attention during the last decade.



The collection of all resonances $\Si (H) \subset \CC$ that are associated with an operator $H$ (in short, resonances of $H$) is 
a \emph{multiset}, i.e., a set in which 
an element $e$ can be repeated a finite number $m_e \in \NN$ of times 
(this number $m_e$ is called the multiplicity of $e$). The \emph{multiplicity of a resonance} $k$ 
is defined as the multiplicity of the corresponding generalized pole of 
$(H-z^2)^{-1}$ (e.g., \cite{DZ17,Z17}) or as the multiplicity of 
a certain analytic function built from the resolvent of $H$ 
and generating resonances as its zeros (\cite{AH84,DEL10,DP11}).

Presently, there are only few Hamiltonians $H$ for those it is known that 
$\Si (H)$ essentially decomposes into sequences 
$\{k_n\}_{n \in \ZZ}$ 
with prescribed asymptotics (`essentially' here and below means that the decomposition takes place 
after a possible exclusion of a finite number of resonances).
Almost all such $H$ are either one-dimensional \cite{K68,Z87,CZ95,S96,P97}
or radial symmetric  (see \cite{Z89JFA,S06,DZ17,Z17} and references therein). 
The exceptions are two  examples 
of 3-dimensional (3-D) Schrödinger Hamiltonians 
with 2 and 3 point interactions that were considered in \cite{AH84} and for those
$\Si (H)$ consists of one and two  sequences, resp., both with logarithmic asymptotics of the form 
$C_1  t  - \ii C_2 \Ln  |t|  + C_0 +   o (1)$ as integer $t$ goes to $\pm \infty$
(where $C_1, C_2>0$ and $C_0 \in \CC$ are constants).
For the arguments of \cite{AH84} concerning 3 point interactions 
the presence of additional symmetries was important, i.e., the example 
possesses the group of symmetries of a regular triangle embedded into 
$\RR^3$.
Here and below under \emph{the group of symmetries of} $Y \subset \RR^3$ we understand the group of 
isometries $\I$ of $\RR^3$ such that $\I Y=Y$.

The first goal of the present paper is to prove that 
$\Si (H)$ essentially decomposes into a finite number of sequences with logarithmic asymptotics 
for an arbitrary point interaction Hamiltonian
$H=H_{a,Y}$ 
corresponding to the formal differential expression 
\begin{equation}\label{e:H}
-\De u (x) + `` \sum_{j=1}^N \mu (a_j) \de (x - y_j) u (x)  ``, \quad x \in \RR^3 , \ N \in \NN, 
\end{equation}
with a finite number $N \ge 2$ of \emph{interaction centers} $y_j \in \RR^3$ and 
the tuple $a=(a_j)_{j=1}^N \in \CC^N$ of \emph{`strength' parameters} (see \cite{AH84,AGHH12,AK00,AK17} 
and Section \ref{s:Def} for the definition of 
$H_{a,Y}$).
The second goal is to connect the leading parameters of these asymptotic sequences 
to the geometry of the set $Y = \{ y_j \}_{j=1}^N$ 
(Section \ref{s:Geom}).

Then we investigate to which extent the result about the structure of $\Si (H_{a,Y})$ can be carried over 
from the case of point interactions 
to the case of quantum graph Hamiltonians (Section \ref{s:G}).

 Our interest in the internal structure of $\Si (H)$ 
is motivated, in particular, by the necessity to consider the notion of
 high-energy resonance asymptotics from the point of view of Physics ,
where only the  resonances  that are closer to $\RR$ play role (see 
\cite{EL17} and Section \ref{s:Physics}),
and by the recent studies of narrow `topological' resonances \cite{GSS13,dVT18}.
The rigorous definition of structural parameters of $\Si (H)$ can also provide an approach to optimization 
problems of the type of \cite[Section 8]{CZ95}, which involve the whole set $\Si (H)$ and are much less studied 
than problems on optimization of an individual resonance \cite{K13,Ka14,KLV17}.

The results on the structure of $\Si (H)$ in the cases of point interactions and quantum graphs 
allow us to consider logarithmic  
$\{z \in \CC \ : \ -C_2 \ln (|\re z|+1) -C_0 \le \im z  \le -C_1 \ln (|\re z|+1) + C_0 \}$ 
or horizontal $\{ -C_2  \le \im z \le - C_1 \}$ strips  that contain `more narrow' resonances,
to define related asymptotic densities, and, in the case of point interaction Hamiltonians,
to connect the density of narrow resonances to the group of symmetries of the set of interaction centers.

Note that various types of strips containing `more narrow' resonances 
have been intensively studied in the context of obstacle 
scattering \cite{SZ93,SZ95,SZ99,Z99,I07,Z17}, 
mainly, with the use of corresponding counting functions.


Recall that the \emph{counting function} $\N_H (\cdot)$ for resonances 
is defined by 
\begin{gather} \label{e:N(R)}
\N_H (R) := \# \{ k \in \Si (H) : |k| \le R \} ,
\end{gather}
where $\#E$ is the number of elements of a multiset $E$.
The study of the asymptotics of $\N (R)$ as $R \to \infty$ for scattering poles associated 
with compactly supported potentials
in $\RR^m$ with odd $m \ge 3$ was initiated in \cite{M83} 
(for the relation between the notions of scattering poles and resonances, 
see \cite{DZ17,Z17}). This study was continued \cite{Z89JFA,F98,SaB01,CH05,S06}
and extended to a geometric scattering (see reviews in \cite{Z99,DZ17,Z17}), 
to quantum graphs \cite{DEL10,DP11,L16,EL17}, and to point interaction Hamiltonians
$H_{a,Y}$ \cite{LL17,AK18}.

For any unbounded set $S \subset \CC$, one can define a restricted version of the counting function 
$ \N_H^{S} (R) := \# \{ k \in S \cup \Si (H) \ : |k|\le R \}  $.
It seems that the study of the asymptotics of $ \N_H^{S} (R)$ 
was originally motivated \cite{I86} by the fact 
that in the case of even $m$, the generalized resolvent has a $\log$-type branching point at $z=0$.
Therefore the counting of resonances was restricted to sectors $\th \le \arg z \le 0$, $\th \in \RR_-$, of 
the associated Riemannian surface \cite{V94_DMJ} (for the case of odd $m$, see \cite{SZ93}). 
On the other side, the existence of infinitely many resonances 
in a certain horizontal strip $S$
and asymptotic sequences in $S$ were studied for the case of two convex obstacles 
 in connection with trapped broken characteristic rays (see \cite{I83,G88,I07} and the
  reviews in \cite{Z99,I07}).
Trapping effects for obstacle scattering have motivated also the investigation of asymptotics 
of $ \N_H^{S} (R)$  for different types of shaped strips $S$, including semi-logarithmic strips 
$\{ -C_1 \ln (|\re z|+1)  \le \im z  \le 0 \}$
and cubic strips $\{ -C_2 |\re z|^{1/3} -C_0 \le \im z \le -C_1 |\re z|^{1/3} +C_0\}$
(see \cite{SZ93,SZ95,SZ99,I07,J15} and references therein; here $C_0, C_1, C_2>0$). 

In \cite{DP11},
it was shown that the set  of resonances $\Si (H_\G)$ 
for a noncompact quantum graphs $\G$ with the Kirchhoffs 
coupling  lie in a certain horizontal strip $S$ and the leading term in the asymptotics  of $ \N_H (R) = 
\N_H^{S} (R) $ was connected with the structure of the graph.  In \cite{DEL10}, 
the study of asymptotics of $ \N_H (R)$ 
was extended to quantum-graphs with general self-adjoint couplings, and 
it was shown that the set of resonances is contained in a union of a horizontal strip with a finite number of logarithmic strips. 

In Section \ref{s:Dis}, we discuss some similarities between the structure of $\Si (H)$ for point interactions and quantum graphs on one hand,
and various cases of obstacle scattering on the other hand.

\subsection{Overview of main results and methods of the paper}
\label{ss:iPI}

Our initial step to find asymptotic sequences inside of $\Si (H)$ is to use the observation of  
\cite{DEL10,DP11,LL17} that the set of resonances is a set of zeroes of a certain exponential 
polynomial $F$ for the cases of noncompact quantum graphs and of point interaction Hamiltonians (see also \cite{L16,EL17,AK17,AK18}).
By the P\'{o}lya-Dickson theorem, the zeroes of exponential polynomials are concentrated 
in a finite number of logarithmic strips (see \cite{BG12,DEL10,L16}; note that 
horizontal strips can be considered as a degenerated type of logarithmic ones \cite{BC63}).

We show that, in the case of a point interaction Hamiltonian $H_{a,Y}$ associated with (\ref{e:H}), 
the corresponding $F$ falls into a special class of exponential polynomials whose distribution 
diagrams generate only retarded logarithmic strips (here and below the terminology of \cite{BC63} is used).
The set of zeroes for such exponential polynomials splits into a finite number of asymptotic sequences.
Translating this result to the (multi-)set $\Si (H_{a,Y})$ we show that
$\Si (H_{a,Y})$ essentially
consists of sequences that for  integer $t $ with large $|t|$ have an asymptotics of the form
\[ 
2 \pi \mu_n t  - \ii \mu_n \Ln | t|  + C +   o (1)
\] 
and whose leading parameters $\mu_n$ 
form a finite set $\{ \mu_n \}_{n=1}^M \subset \RR_+$.
It is natural to say that $\mu_n$ are the \emph{structural parameters of 1st order}
(\emph{the parameters of 2nd order} $\om_{n,j}$ are hidden inside constants $C$, see Theorem \ref{t:ra} for details).
Each of the parameters $\mu_n$ participates in a number $r_n \in \NN $ of asymptotic sequences. It is natural to call the number $r_n$
the \emph{multiplicity of $\mu_n$}.

To show in Theorems \ref{t:gen}-\ref{t:MaxStr} the interplay between  $\{ \mu_n \}_{n=1}^M$ and the metric geometry of the tuple $Y$, 
we introduce a sequence of values $\Size_m$, which are called \emph{$m$-sizes of $Y$} and are generalizations 
of the size of $Y$ introduced in \cite{LL17}.  For integer $m \in [0,N]$, 
let $S_{N,m}$ be the set of permutations $\si \in S_N$ of the symmetric group $S_N$ such that $\si $ 
has exactly $N-m$ fixed points,
i.e., $
S_{N,m} := \{ \si \in S_N \ : \ \#\{i : \si (i) \neq i \} = m\}.
$
For integer $m \in \{0\} \cup [2,N]$, we define \emph{the  $m$-size of $Y$} by
\begin{align}
\Size_m & = \Size_m (Y) :=  \max_{\si \in S_{N,m}} \sum_{j=1}^N | y_{j} -  y_{\si (j)}| ;  \notag \\
\Size_1 & = \Size_1 (Y) := \diam Y , \text{ where $\diam (Y) := \max_{1 \le j,j'\le N} |y_j - y_{j'}|$ is the diameter of $Y$.}
\notag
\end{align}
Then $\Size_0 = 0$, $\Size_2 = 2 \diam Y$, 
$\Size_3$ is the maximal perimeter of a triangle with vertices in $Y$;
$1$-size of $Y$ is defined in a special way since $S_{N,1} = \varnothing$.
The logic behind this is that the equality $\Size_1 := (\Size_2 + \Size_0)/2$ simplifies the formulation of Theorem \ref{t:MaxStr}.
The $N$-size $\Size_N$ was called in \cite{LL17} the size of $Y$ and used to define Weyl and non-Weyl types 
of asymptototics for $\N_{H_{a,Y}} (\cdot)$ (see the beginning of Section \ref{s:Geom}).

The connection between the parameters $\mu_n$ and $m$-sizes $\Size_m$ is established 
in Theorems \ref{t:gen}-\ref{t:MaxStr}.
However it is difficult to give an explicit rule of computation of the leading parameters 
$\mu_n$ via $m$-sizes that works for any $Y$. The reasons for this are shown by the examples of Section \ref{ss:ex}.

The structural theorem for $\Si (H_{a,Y})$ in the case of point interactions is quite special and 
cannot be directly brought  over even to the case of quantum graphs,
which is considered in Section \ref{s:G}. However this theorem give a hint how one can 
define the structural parameters even if the decomposition of $\Si (H)$
into asymptotic sequences is not available. With this aim we use counting functions in `shaped strips'
similar to that of \cite{SZ93,SZ95} and define via them corresponding asymptotic 
density functions. The main examples of such functions in this paper 
are the \emph{logarithmic counting functions}
\begin{align} \label{e:Nlog}
\N^{\log} (\mu,R) := \# \{ k \in \Si (H) \ : \ - \mu \ln (|\re k|+1) \le  \im k \text{ and } |k|\le R \} , 
\end{align}
and the associated \emph{logarithmic density function} 
\begin{align} \label{e:AdLog}
\d^{\log} (\mu) := \lim_{R\to \infty} \frac{\N^{\log} (\mu,R)}{R} , \quad \mu \in \RR , \text{ and }
\d^{\log} (+\infty) := \lim_{R\to \infty} {\N_H (R)/R} .
\end{align}
The above definition of $\d^{\log} (+\infty)$ is natural because, for $H_{a,Y}$ and 
for quantum graphs,  this limit exists and has 
a finite value that was studied in \cite{DP11,DEL10,LL17,AK18} 
in connection with Weyl-type asymptotics of $\N_H (\cdot)$.

For $H_{a,Y}$ and for quantum graphs, the function
$\d^{\log} : (-\infty,+\infty] \to [0,+\infty) $ is  bounded, nondecreasing and satisfies $\d^{\log} (\mu) =0$
for $\mu < 0$ and $\d^{\log} (\mu) = \d^{\log} (+\infty)$ for large enough $\mu$
(this follows from the definition, results of \cite{DEL10,DP11,LL17}, and Theorem \ref{t:ra}).
So it defines a bounded measure $d \d^{\log} (\cdot)$ on $\RR$ with a compact support.
The connection of $\d^{\log} (\cdot)$ with the structure of $\Si (H)$ is obvious in the case of
$H_{a,Y}$ from Theorem \ref{t:ra}  and in the case of quantum graphs from Theorem \ref{t:KSH} 
(the latter evolves from \cite[Theorem 3.1]{DEL10}, but describes $\Si (H)$ on a more fine structural level).
Namely, the measure $d \d^{\log} (\cdot)$ consists of a finite number of point masses, 
the minimum $\mu^{\min} $ of its support and the height of the corresponding jump $\d^{\log} (\mu^{\min}+0) - \d^{\log} (\mu^{\min}-0)$
are the parameters describing the high-energy asymptotics of `most narrow' (and so most physically relevant)
resonances (see Section \ref{s:Physics} and the discussion in \cite{E84}).

For $H_{a,Y}$, the parameter $\mu^{\min}$ 
is always equal $1/\diam Y$.
So the main parameter of the formula (\ref{e:anar}) for \emph{the asymptotic density of narrow resonances} 
$\d^{\log} (\mu^{\min}+0) - \d^{\log} (\mu^{\min}-0)$ is the multiplicity $r^\narrow$ of 
$\mu^{\min}$, which is an integer number  between $2$ and $N$ and is studied by Theorems \ref{t:F2} and \ref{t:rM>2}.
These theorems naturally lead to a conjecture about the connection of the value of $r^\narrow$ with the group of symmetries 
of $Y$.

The structural description of $\Si (H)$ in the case of a quantum graph Hamiltonian $H$ is given in Section \ref{ss:QG}. 
It is more complex than that of $\Si (H_{a,Y})$. In the case where $\mu^{\min}=0$, 
a certain neutral strip $\{ z \in \CC : |\im z| \le \wt \ga\}$ contains 
an infinite number of resonances, that are not necessarily  decomposable into a finite number of asymptotic sequences
(at least we do not know a general enough theorem of such type for exponential polynomials).
In the important case of the Kirchhoff coupling, the situation is the worst possible. The measure $d \d^{\log} (\cdot)$ 
consists of one point mass at $\mu^{\min}=0$, i.e., there exists only one 1st order structural parameter, 
which does not describe any structure. The description of the structure of $\Si (H)$
 should be encoded in the measure $d \d^\hor$ associated with another density function  (\ref{e:AdHor})
built with the use of horizontal strips. However, the distribution of resonances in the non-decomposable cases 
is connected with difficult questions arising in the exponential polynomial approach to the study of zeros of 
Riemann zeta function (see the monograph \cite{BG12} and the references therein).
To obtain a complete decomposition of $\Si (H)$ into asymptotic sequences we impose the additional 
commensurability condition $\ell_{m_1}/\ell_{m_2} \in \QQ$ on the lengths $\ell_m$ of edges of the graph $\G$ 
(see Theorem \ref{t:Kir} and the part (iii) of Theorem \ref{t:KSH}, here and below $\QQ$ is the set of 
rational numbers).

In Section \ref{ss:PhC}, we brought the above results over the set of 
resonances of 1-D photonic crystals using the fact that they can be considered as generalized `weighted'
quantum graphs with the Kirchhoff coupling (for another type of coupling for weighted graphs, see \cite{DEL10}).

\textbf{Notation}. 
We use the convention that if $n_2 < n_1$, then $\{ z_n \}_{n=n_1}^{n_2} = \varnothing$.
The following standard sets are used: the sets $\NN$, $\ZZ$, $\QQ$, and $\RR$ of natural, integer, rational, and real 
numbers, resp., the lower 
and upper 
complex half-planes 
$\CC_\pm = \{ z \in \CC : \pm \im z > 0  \}$,
open half-lines $\RR_\pm = \{ x \in \RR: \pm x >0 \}$,
open discs $\DD_\vep (\zeta) := \{z \in \CC : | z - \zeta | < \vep \}$,
compact $\CC$-intervals $[z_1,z_2] := \{ s z_1  + (1-s) z_2 \ : \ s \in [0,1] \}$ with $z_1, z_2 \in \CC$,
$\CC$-intervals $(z_1,z_2) := [z_1,z_2] \setminus \{ z_1, z_2\}$ without endpoints.
The above $\CC$-intervals are called degenerate if $z_1 =z_2$.
In a metric space $U$ with the distance function $\rho_U (\cdot,\cdot)$ (or in a normed space),
we use open balls 
$
\BB_\vep (u_0) := \{u \in U \, : \, \rho_U (u, u_0) < \vep \}
$ always assuming that $\vep>0$.
For  a normed space $U$, $u_0 \in U$, $S \subset U$, and $z \in \CC$, we write 
$
z S  + u_0 := \{ zu + u_0 \, : \, u \in S \}
$. 
For a function $g$ defined on $S$, $g[S]$ is the image of $S$.
The function $\Ln (\cdot)$ ($\Arg_0 (\cdot)$) is the branch of the natural logarithm multi-function $\ln (\cdot)$ 
(resp., the  complex argument $\Arg (\cdot)$) in $\CC \setminus (-\infty,0]$ fixed by  $\Ln 1 = \ii \Arg_0 1 = 0$. 
For $z \in \RR_-$, we put $\Ln z = \Ln |z| + \ii \Arg_0 z = \Ln |z| + \ii \pi $.
By  $\pa_x f $, $\pa_{x_j} f$, etc., we denote (ordinary or partial)
derivatives  with respect to (w.r.t.) $x$, $x_j$, etc.;
$\deg p$ stands for the degree of a polynomial $p$,
$\deg X_j$ for the degree of a vertex $X_j$ of a graph;
$\id = [1] [2] \dots [N]$ for the identity permutation.
Here and below we use the square brackets notation of the textbook \cite{Lang02}
for permutation cycles, omitting sometimes, when it is convenient, 
the degenerate cycles consisting of one element.


\section{Resonances of point interaction Hamiltonians}
\label{s:Def}

Throughout this paper, the set $Y=\{ y_j \}_{j =1}^N$ consist of $N\ge 2$ distinct points 
$y_1$, \dots, $y_N$ in $\RR^3$. Let $a = (a_j)_{j=1}^N \in \CC^N$ be the $N$-tuple  of 
the \emph{`strength' parameters}. 

The operator $H_{a,Y}$ associated with (\ref{e:H}), where $\delta (\cdot - y_j)$ is the Dirac measure 
placed at the \emph{center} $y_j \in \RR^3$ of a \emph{point interaction}, 
is defined in \cite{AGH82,AGHH12} for the case 
of real $a_j$, and in \cite{AGHS83,AK17} for $a_j \in \CC$.
It is a closed operator in the complex Hilbert space $L^2_\CC (\RR^3)$ 
 and it has a nonempty resolvent set.
The spectrum of $H_{a,Y}$ consists of the essential spectrum $[0,+\infty)$ and an 
at most finite set 
of points outside of $[0,+\infty)$ \cite{AGHH12,AK17} (all of those points are eigenvalues).

The resolvent $(H_{a,Y}-z^2)^{-1}$ of $H_{a,Y}$ is defined in the classical sense 
on the set of $ z \in \CC_+ $ such that $z^2$ is not in the spectrum, 
and has the integral kernel 
\begin{gather} \label{e:Res}
(H_{a,Y}-z^2)^{-1} (x,x')  = G_z (x-x') + \sum_{j,j' = 1}^N G_z (x-y_j)
\left[ \Ga_{a,Y} \right]_{j,j'}^{-1} G_z (x' - y_{j'} ) , 
\end{gather}
where $x,x' \in \RR^3 \setminus Y$ and $x \neq x'  $, see e.g. \cite{AGHH12,AK17}.
Here 
$
G_z (x-x') := \frac{e^{\ii z |x-x'|}}{4 \pi |x-x'|}
$ 
is the integral kernel associated with
the resolvent $(-\De - z^2)^{-1}$ of 
the kinetic energy Hamiltonian $-\De$;
$\left[ \Ga_{a,Y} \right]_{j,j'}^{-1}$ denotes the $j,j'$-element of the inverse to 
the matrix 
\begin{gather} \label{e:Ga}
\Ga_{a,Y} (z) = \left[ \left( a_j - \tfrac{\ii z}{4 \pi} \right) \de_{jj'} 
- \wt G_z  (y_j-y_{j'})\right]_{j,j'=1}^{N}, \text{ where }
\wt G_z (x) := \left\{ \begin{array}{rr} G_z (x), & 
x \neq 0 \\
0 , & 
x = 0  \end{array} \right.  .
\end{gather}

The Krein-type  formula (\ref{e:Res}) for 
the difference of the perturbed and unperturbed resolvents of operators $H_{a,Y}$ and $-\De$ can be used as 
a definition of $H_{a,Y}$ (see \cite{AGHH12}).
For other equivalent definitions of $H$ and for the meaning of $\mu (a_j) $ and $a_j$
in (\ref{e:H}), we refer to 
\cite{AGH82,AGHH12,AK00} in the case $a_j \in \RR$, and to 
$\cite{AGHS83,AK17}$ in the case $a_j \not \in \RR$. 
Note that, in the case $a \in  \RR^N$, 
the operator $H_{a,Y} $ is self-adjoint in $L^2_\CC (\RR^3)$; and 
in the case $a \in  (\CC_- \cup \RR)^N $, 
$H_{a,Y}$ is closed and maximal dissipative (in the sense of \cite{E12}, or in the sense that $\ii H_{a,Y}$ is maximal accretive).

The set of (continuation) resonances $\Si (H_{a,Y})$ associated with the operator $H_{a,Y}$ 
(in short, resonances of $H_{a,Y}$) is by definition the set of zeroes of 
the determinant 
$
\det \Ga_{a,Y} (\cdot)  ,
$
which we will call the characteristic determinant.
This is in agreement with the cut-off resolvent pole definition of \cite{DZ17} and slightly differs 
from the one used in \cite{AH84,AGHH12}
because isolated eigenvalues are now 
also included into $\Si (H_{a,Y})$.
For the origin of this and related approaches to the understanding of resonances, we refer
to \cite{AH84,DZ17,RSIV78} and the literature therein. 
The multiplicity of a resonance $k$ will be understood as the multiplicity 
of a corresponding zero of the determinant 
$\det \Ga_{a,Y} $, which is an analytic function in $z$ (see \cite{AGHH12}).
Equipped with the multiplicity, the set $\Si (H_{a,Y})$ becomes a multiset (for the discussions on multiplicities of resonances
see e.g. \cite{CZ95,DZ17,Ka14,Z17}).

The function $\det \Ga_{a,Y} (\cdot) $ is an exponential polynomial, 
which after a simple transformation becomes of a special type considered in \cite{BC63} (for the general theory see \cite{BG12}).
Namely, introducing a new variable $\zeta = - \ii z$ and denoting $A_j := 4 \pi a_j$,
one can see that the modified characteristic determinant
$D (\zeta) := (-4 \pi)^N \det \Ga_{a,Y} (\ii \zeta) $
can be expanded by the Leibniz formula into the sum of terms
\begin{equation} \label{e:Dterms}
e^{\zeta \al (\si)} p^{[\si]} (\zeta)
\end{equation}
taken over all permutations $\si$ in the symmetric group $S_N$,
Here the constants $\al (\si) \le 0$  and the polynomials $p^{[\si]} (\cdot)$ 
have the form 
\begin{align} \label{e:al si}
\al (\si) := - \sum_{j : \si (j) \neq j  } |y_j - y_{\si (j)}| , \quad 
p^{[\si]} (\zeta) := \ep_\si K_1 (\si) \prod_{j : \si (j) = j  } (-\zeta-A_j) , \\
\text{ where $K_1 (\si) := \prod_{j : \si (j) \neq j  } |y_j - y_{\si (j)}|^{-1} >0$
(in the case $\si = \id$, $K_1 (\id) := 1$),}
\end{align}
$\ep_\si$ is the permutation sign (the Levi-Civita symbol), and $\id$ is the identity permutation
(note that $K_1 (\si)$ depends also on $Y$, and $p^{[\si]} (\zeta)$ on $a$ and $Y$).

We will say that an exponential polynomial $g (\cdot)$ is an exp-monomial if it has 
the form $e^{\beta \zeta} p (\zeta)$, where $\beta \in \CC$ and $p$ is a polynomial that is nontrivial in the sense that
$p (\cdot) \not \equiv 0$.

\section{Logarithmic asymptotic chains of resonances}
\label{s:Asy}

\subsection{The distribution diagram and logarithmic strips}
\label{ss:L}

Summing (\ref{e:Dterms}) and writing the exponential polynomial $D$ 
in the canonical form \cite{BG12,BC63}, one obtains 
\begin{equation} \label{e:CanForm}
D(\zeta) = \sum_{j=0}^\nu P_{\beta_j} (\zeta) e^{\beta_j \zeta } ,
\end{equation}
where $\nu \in \NN \cup \{0\}$, 
$\beta_j \le 0$, and the nontrivial polynomials $P_{\beta_j} $  
(with coefficient depending on $a$ and $Y$) are such that 
$ \beta_0 < \beta_1 < \dots < \beta_\nu = 0$.
Clearly, 
$
P_0 (z) = P_{\beta_\nu} = \prod_{j=1}^N  (-\zeta - A_j ) .
$ 
The coefficients $\beta_j$ in (\ref{e:CanForm}) (and $\al (\si)$ in (\ref{e:Dterms}))
are called frequencies of the corresponding exponential polynomials (resp., exp-monomials).
We will use the convention that 
\begin{align} \label{e:Pb}
\text{if $b \not \in \{ \beta_j \}_{j=0}^\nu$, then $P_b (\cdot)$ is trivial in the sense that $P_b (\cdot) \equiv 0$.}
\end{align}

Since we rely on the terminology and the results of the theory of zeros of exponential polynomials 
given through \cite[Sections 12.4-8]{BC63}, we try to keep our notation as close as possible to that of \cite{BC63}.
It is difficult to achieve this aim completely, in particular, our frequencies $\beta_j$ are nonpositive, 
while those of \cite{BC63} are nonnegative.

It is obvious from the definition of $\Size_m$ (see Section \ref{s:intro}) that
\begin{multline} \label{e:al=-Pi}
\text{for each $m \in (\ZZ \cap [0,N])\setminus \{1\} $,} \\
\text{there exists $\si \in S_N$ such that $\al(\si) = (-\Size_m)$ and $\deg p^{[\si]} = N-m$}.
\end{multline}

In the process of summation of exp-monomials of (\ref{e:Dterms}) some of the 
terms  may cancel so that, 
for a certain permutation $\si \in S_N$, $\al (\si)$ is not a frequency of $D $.
If this is the case, we say that there is \emph{frequency cancellation} for the pair 
$\{a,Y\}$
(for two examples of frequency cancellation see \cite{LL17} and Example \ref{ex:Y4}).

Since $N \ge 2$, there exists a nonzero frequency 
$\al (\si)$ that does not cancel for $D$, i.e., $\nu \ge 1 $ and $\al (\si)$ 
is a frequency of $D$ for certain  $\si \neq \id $.
This fact was observed in \cite[the proof of Lemma 2.1]{AK17} and, by a different argument,
can also  be seen from the next lemma.

\begin{lem} \label{l:Pi2Pi3}
(i) The number $(-\Size_2)$  is a frequency of $D (\cdot)$. Besides,  
$\deg P_{-\Size_2} = N-2$.

\noindent (ii) Let $N \ge 3$. Then $(-\Size_3)$ is a frequency of $D (\cdot)$. Moreover, 
\begin{align} \label{e:S3oneL}
\text{$\Size_3 = \Size_2$ if all the points of the set $Y$ lie on one line; otherwise, $\Size_3 > \Size_2$.}
\end{align}
\end{lem}

\begin{proof}
\emph{(i)} 
Consider the class of all transpositions $\si = [j_1 j_2]$, $j_1 \neq j_2$, such that 
\[
 |y_{j_1} - y_{j_2}| = \diam Y = \Size_2/2 .
 \]
 The corresponding terms of (\ref{e:Dterms})
are  
$(-1) e^{- \Size_2 \zeta} (\diam Y )^{-2} \prod_{j \neq j_1,j_2} (-\zeta - A_j)$. 
The highest order coefficients of their polynomial factors $p^{[\si]} $
coincide and cannot cancel each other. The exp-monomial 
produced by summation of these terms has the polynomial part of degree $N-2$ and 
cannot be canceled by other terms of 
(\ref{e:Dterms}) because the other terms either have a lower $\deg p^{[\si]}$, or a different frequency $\al (\si)$. Thus, $(-\Size_2)$ is a frequency of $D$. 

\emph{(ii)} The statement (\ref{e:S3oneL}) is obvious from the triangle inequality. In the case $\Size_3 = \Size_2$, $(-\Size_2)$ is a frequency of $D$
due to (i). Assume $\Size_3 > \Size_2$. Then it is easy to modify the proof of (i) to show 
that the frequency $(-\Size_3)$ does not cancel.
\end{proof}


Consider the points 
$
T_{\beta_j} = \beta_j + \ii \deg P_{\beta_j} \in \CC , \quad j=0,\dots,\nu, 
$
associated 
with the canonical form (\ref{e:CanForm}) of $D$.
Note that $T_{\beta_\nu} = T_0 = \ii N$. 

\emph{The distribution diagram of $D$} (see \cite{BC63}) is the polygonal line $\L$ in $\CC$ 
determined by the following properties:
\begin{itemize}
\item[(L1)] $\L$ joins $T_{\beta_0}$ with $T_{\beta_\nu}$.
\item[(L2)] $\L$ has vertices only at the points of the set $\{T_{\beta_j}\}_{j=1}^{\nu}$.
\item[(L3)] $\L$ is convex upward. (In particular, it is allowed to be of the form of one $\CC$-interval 
$[z_0,z_1]$. This is the case for $N=2$, but not only, see Section \ref{ss:ex}.)
\item[(L4)] There are no points of the set $\{T_{\beta_j}\}_{j=1}^{\nu}$ above $\L$ in the sense  that 
the $\CC$-intervals $(T_{\beta_j}, \beta_j - \ii )$ (with excluded endpoints)  do not intersect $\L$.
\end{itemize}

Let $\L_1$, $\L_2$, \dots, $\L_M$ be the successive segments of $\L$  numbered from
left to right and such that $\L = \cup_{n=1}^M \L_n$ 
(it is assumed that each segment $\L_n$ is a closed nondegenerate $\CC$-interval and 
the segments are maximal in the sense 
that two consecutive segments do not belong to the same line).
Note that $\L$ has no vertical segments (lying on the lines $\re z = c $) due to (L4) 
(see also \cite[Section 12.8]{BC63}).

Let 
$\{T_{n,h}\}_{h=1}^{\nu_n} := \L_n \cap \{ T_{\beta_j} \}_{j=0}^\nu$,
where $\L_n = [T_{n,1},T_{n,\nu_n}]$ and the points $T_{n,h}$ on the segment $\L_n$ 
are numbered from left 
to right, i.e., $\re T_{n,1} < \dots < \re T_{n,\nu_n} $,
\begin{gather} \label{e:T11MnuM}
\text{$T_{1,1} = T_{\beta_0}$, \quad $T_{M,\nu_{_M}} = T_0 = \ii N$, \quad 
and, if $2 \le n \le M$, 
$T_{n-1,\nu_{n-1}} = T_{n,1}$.}
\end{gather}

Put  
$\beta_{n,h} := \re T_{n,h}$ and $m_{n,h} := \im T_{n,h} = \deg P_{\beta_{n,h}}$.
Then the slopes $\mu_n$ of the segments $\L_n$ are defined by 
\begin{align} \label{e:mun}
\mu_n := \tan ( \arg (T_{n,\nu_n} -  T_{n,1})) = 
\frac{m_{n,\nu_n} - m_{n,1} }{\beta_{n,\nu_n} - \beta_{n,1}} ,
\qquad n= 1, \dots, M.
\end{align}
Note that 
\begin{gather} \label{e:mu decreasing}
\text{the sequence $\{\mu_j\}_{j=1}^M$ is strictly decreasing.}
\end{gather}
This follows from the convexity (L3) and the definition of 
the segments $\L_n$.

The P\'{o}lya-Dickson theorem on the zeroes of exponential polynomials 
\cite{BC63,BG12} states that there exists constants $w_n \ge 0$, $n=1,\dots,M,$ and 
the logarithmic strips 
\begin{equation} \label{Vmuw}
\V (\mu_n, w_n) := \{ \zeta \in \CC \ : \ |  \re (\zeta  + \mu_n \ln \zeta)  | \ \le \ w_n \}  , \quad 
n=1,\dots,M,
\end{equation}
such that $\CC \setminus \bigcup_{j=1}^{M} \V (\mu_n,w_n) $ can be decomposed 
into a finite number of subsets with the property that 
in each of them one of the sums $\sum_{T_j \in \L_n} P_{\beta_j} (\zeta) e^{\beta_j \zeta}$ 
is of predominant order of magnitude 
over the other terms in (\ref{e:CanForm}) as $\zeta \to \infty$ \cite{BC63}.
In particular, this implies the following statement.

\begin{lem}[\cite{BC63}] \label{l:C-V}
There is only a finite number of zeroes of $D(\cdot)$ in $\CC \setminus \bigcup_{j=1}^{n} \V (\mu_n,w_n) $.
\end{lem}

\subsection{Bunches of asymptotics chains of resonances}
\label{ss:ra}

The main point of this subsection is that the modified characteristic determinant $D$ belongs to 
the more special subclass of exponential polynomials that generate only positive parameters $\mu_n $,
i.e., in the terminology of \cite{BC63}, all the logarithmic strips $ \V (\mu_n, w_n) $ for $D$ are \emph{retarded}.
This allows us to apply the results of \cite{BC63} about this subclass and in this way to 
split $\Si (H_{a,Y})$ into a finite number of sequences with a prescribed form of 
asymptotics at $\infty$.

The fact that $\{ \mu_n \}_{n=1}^M \subset \RR_+$ follows from (\ref{e:mu decreasing}) and 
the following statement.

\begin{prop} \label{p:muM} 
$\mu_M =  1/\diam Y = 2/\Size_2$ and  $1 \le M \le N-1$.
\end{prop}
\begin{proof}
Let us show that $\mu_M =   2/\Size_2$.
It follows from Lemma \ref{l:Pi2Pi3} that $(-\Size_2)$ is a frequency of $D$. Further,
the arguments of the proof of Lemma \ref{l:Pi2Pi3} show that 
\begin{gather} \label{e:PJPi2}
\deg P_{-\Size_2} = N-2, \quad T_{-\Size_2} = -\Size_2 + \ii (N-2) .
\end{gather}
The convexity (L3), (\ref{e:T11MnuM}), and the definition of $\mu_n$ imply that 
$
\mu_M \le \frac{\im ( T_{M,\nu_{_M}} - T_{-\Size_2})}{\re (T_{M,\nu_{_M}} - T_{-\Size_2})} = 
\frac{2}{\Size_2} .
$

Let us prove that $\mu_M \ge 1/\diam Y = 2/\Size_2$ by \emph{reductio ad absurdum}.
Assume $\mu_M < 1/\diam Y$. Then there exist $j < \nu$ such that 
$
\frac{N - \deg P_{\beta_j}}{-\beta_j} < \frac{1}{\diam Y} .
$
So there exists $\si \in S_N$ such that $\al (\si) = \beta_j$ and 
$(N- \deg p^{[\si]}) \diam Y< -\al (\si)$. However,
it follows from (\ref{e:al si}) that 
$(-1) \al (\si) 
\le \sum_{j : \si (j) \neq j  } \diam Y =
(N- \deg p^{[\si]}) \diam Y $,
a contradiction.

By Lemma \ref{l:Pi2Pi3}, $M \ge 1$. Let us show that $M\le N-1$. We see that 
\begin{gather} \label{e:m_nh inZ}
m_{n,h} = \im T_{n,h} = \deg P_{\beta_{n,h}} \in (\ZZ \cap [0,N])\setminus \{N-1\}
\end{gather}
and, from the definitions of $T_{n,h}$, (\ref{e:mu decreasing}), and $\mu_n>0$, that the sequence $\im T_{1,1}$, $\im T_{2,1}$, \dots, $\im T_{M,1}$, $\im T_{M,\nu_M}$ is strictly increasing.
This implies that $\L$ has at most $N-1$ segments. 
\end{proof}

The following theorem is the main result of the section.

\begin{thm} \label{t:ra}
There exist numbers $M \in \NN $, $M_0 \in \NN \cup \{0\}$, finite sequences 
$\{ \mu_n \}_{n=1}^M \subset \RR_+$, $\{r_n\}_{n=1}^M \subset \NN$, 
$\{ \om_{n,j} \}_{j=1}^{r_n} \subset \CC \setminus \{0\}$ and $\{ t_{n,j} \}_{j=1}^{r_n} \subset \NN$ for $n=1,\dots, M$,
$\K_0 = \{k^{[0]}_t \}_{t=1}^{M_0} \subset \CC$, and infinite sequences 
$\K_{n,j}^\pm = \{k_{n,j,t}^\pm\}_{t=t_{n,j}}^{+\infty} \subset \CC$ such that:
\item[(i)]  $\Si (H_{a,Y}) = \bigcup_{n=0}^M \K_n $ (taking into account multiplicities), 
where 
\[ 
\textstyle \text{$\K_n := \left(\bigcup_{j=1}^{r_n} \K_{n,j}^- \right) \ \cup \ \left(
\bigcup_{j=1}^{r_n} \K_{n,j}^+ \right) $, $n=1$, \dots, $M$.}
\]
\item[(ii)] Each of $\K_{n,j}^\pm $ has the following asymptotics as $t \in \NN$ goes to $+\infty$:
\begin{gather} \label{e:k njt}
\frac{k_{n,j,t}^\pm}{\mu_n} = 
\pm 2 \pi t  - \ii \Ln (t)  \mp \pi/2  - \ii \Ln (2 \pi \mu_n) + \ii \Ln (\om_{n,j}) + 
  o (1)  .
\end{gather}

\item[(iii)] The sequence $\{ \mu_n \}_{n=1}^M $ is strictly decreasing
and $\mu_M = \tfrac1{\diam Y}$.
\item[(iv)] $1 \le M \le N-1$ and $2 \le r_M \le \sum_{n=1}^M r_n \le N $.
\end{thm}

The proof given below also explains how 
the parameters of the asymptotics (\ref{e:k njt}) can be found from the 
distribution diagram $\L$. 

\begin{proof} 
We use the notation of Section \ref{ss:L} and the parameters defined therein.
In particular, $M \ge 1$ is the number of segments of $\L$, 
$\{ \mu_n \}_{n=1}^M $ is the sequence defined by (\ref{e:mun}).
Let $c_{n,h}$ be the coefficient corresponding to the monomial of the highest degree for the polynomial 
$P_{\beta_{n,h}}$, i.e., 
$c_{n,h} = \lim_{\zeta \to \infty} P_{\beta_{n,h}} (\zeta)/\zeta^{m_{n,h}} \neq 0$,
$h=1, \dots, \nu_n$.
Consider the polynomial 
$q_n (\om) := \sum_{h=1}^{\nu_n} c_{n,h} \om^{m_{n,h}-m_{n,1}}$
of the degree
\[
r_n:=m_{n,\nu_n}-m_{n,1}
\]
 and the finite sequence of its zeroes $\{ \om_{n,j} \}_{j=1}^{r_n}$ 
(repeated according to their multiplicities). Note that $c_{n,1} \neq 0$  imply 
$\om_{n,j} \neq 0 $.

The statement (iii) of the theorem is already proved. Combining  (iii) with 
 \cite[Theorem 12.8 and 12.10 (d)]{BC63} (see also \cite[(12.8.11-12)]{BC63}, \cite[Theorems 12.6-8]{BC63},
and \cite[Lemma 12.4]{BC63} for details), we see that for each $n=1$, \dots, $M$, there exists 
$R_n \ge 0$ and sequences $\{ \zeta_{n,j,t}^\pm \}_{t=t_{n,j}^\pm}^{+\infty} \subset \CC_\pm$,
$j=1, \dots,  r_n$ with the starting numbers 
$t_{n,j}^\pm \subset \NN$ satisfying the following properties:
\begin{itemize}
\item[(a)] The sequence $\Z_n^\pm$ of all the zeroes of $D$ in 
$(\V (\mu_n,w_n) \cap \CC_\pm) \setminus \DD_{R_n} (0)$ 
 is the union of $r_n$ sequences 
$\{ \zeta_{n,j,t}^\pm \}_{t=t_{n,j}^\pm}^{+\infty}$ (taking into account multiplicities).
\item[(b)] Each of the sequences $\{ \zeta_{n,j,t}^\pm \}_{t=t_{n,j}^\pm}^{+\infty}$ has the asymptotics
(cf. \cite[Theorem 12.8]{BC63})
\begin{gather*} \label{e:Zeta n}
\frac{\zeta_{n,j,t}^\pm}{\mu_n} = \Ln \om_{n,j} - 
\Ln | \pm 2 \pi \mu_n t + \mu_n \Arg_0 (\om_{n,j})  \mp \mu_n \pi/2 | + 
\ii (\pm 2 \pi t \mp \pi/2) + o (1) .
\end{gather*}
\end{itemize}
 
Due to Lemma \ref{l:C-V},
one sees that only a finite sequence $\Z_0$ of zeroes of $D$ does not get into the multiset
$\bigcup_{j=1}^{n} \Z_n^\pm$. Passing from $\zeta$ to $z= \ii \zeta$ and from zeroes of $D(\zeta)$ to 
that of $\det \Ga_{a,Y} (z)$, we obtain statements (i)-(ii) of Theorem \ref{t:ra}.
Note that $\ii \zeta_{n,j,t}^\pm = k_{n,j,t}^\mp$.

The part $1 \le M \le N-1$ of statement (iv) is proved by Proposition \ref{p:muM}. To prove 
$\sum_{n=1}^M r_n \le N$, note that 
(\ref{e:m_nh inZ}) and  (\ref{e:T11MnuM}) imply
$\sum_{n=1}^M r_n = N - m_{1,1} \le N $. The fact that $r_M \ge 2$ follows from Lemma \ref{l:Pi2Pi3} (i)
and Proposition \ref{p:muM}.
This completes the proof of Theorem \ref{t:ra}.
\end{proof}

\section{Geometry of $Y$ and parameters of asymptotics}

\label{s:Geom}

The parameters $\mu_n$ play the role of leading parameters  of the asymptotic sequences 
$\{k_{n,j,t}^\pm\}_{t=t_{n,j}}^{+\infty}$ of (\ref{e:k njt}). Their role in the distribution of resonances 
is emphasized by the next statement, which is immediate corollary of Theorem \ref{t:ra} 
(see also \cite[formula (12.18.16)]{BC63} and the correcting remark to it in the Russian edition of \cite{BC63}).

\begin{cor} \label{c:dHaY}
 For the operator $H_{a,Y}$, the logarithmic asymptotic density function $\d^{\log} (\cdot)$ defined by (\ref{e:AdLog})
is a nondecreasing piecewise constant function with a finite number of jumps. These jumps are exactly  at the points of the set 
$\{ \mu_n \}_{n=1}^M \subset \RR_+$,
and the height of each jump is equal to 
\begin{align} \label{e:jumps}
\d^{\log} (\mu_n+0) - \d^{\log} (\mu_n-0) = \frac{\beta_{n,\nu_n} - \beta_{n,1}}{\pi} = 
\frac{r_n}{\pi \mu_n }.
\end{align}
\end{cor}

For the following considerations let us recall that the $m$-sizes $\Size_m$ of $Y$ were defined in Section \ref{s:intro} and that the value $\Size_N$ was introduced 
in \cite{LL17} and called the \emph{size of $Y$}. Let $\N_{H_{a,Y}} (\cdot)$ be the resonance counting function defined 
by (\ref{e:N(R)}) with  $H=H_{a,Y}$.

It was shown in \cite[Theorem 4.1]{LL17} (see also comments in \cite{AK18} about the choice of the `strength' tuple $a$) that 
$\N_{H_{a,Y}} (R) = \frac{W(a,Y)}{\pi} R + O(1)$, where $0 \le W(a,Y) \le \Size_N$, and it was said 
that $\N_{H_{a,Y}} (\cdot) $ has the \emph{Weyl-type asymptotics} if $W(a,Y) = \Size_N$ (i.e., roughly speaking, when
$W(a,Y)$ attains the maximal possible value). 
Slightly rephrasing \cite{LL17}, it is natural to say that $W(a,Y)$ is the \emph{effective size of $H_{a,Y}$}.

Taking into account Theorem \ref{t:ra}, we see that
\[
\lim_{\mu \to +\infty} \d^{\log} (\mu) = \d^{\log} (+\infty) = W(a,Y)/\pi \le \Size_N/\pi.
\]
With the use of Theorem \ref{t:ra} (iii)-(iv) one can strengthen \cite[Theorem 4.1]{LL17} by the following inequality
\[
W(a,Y)/\pi = \d^{\log} (+\infty)  \ge \d^{\log} (\mu_M+0) =  
\frac{r_{\scriptscriptstyle M}}{\pi \mu_{\scriptscriptstyle M} } \ge \frac{2 \diam Y}{\pi } .
\]
Thus, $W(a,Y) \ge 2 \diam Y = \Size_2$.  The following statement concerns
configurations of $Y$ whose  
effective sizes $W(a,Y)$ take the minimal possible value $\Size_2$.

\begin{cor} \label{c:w=2d}
If $W(a,Y) = 2 \diam Y $, then all the points of $Y$ lie on one line.
\end{cor}
\begin{proof}
If $Y$ does not belong to one line, then Lemma 3.1 (ii) implies that $\Size_3>\Size_2$ and that 
$\beta_0 $ in (\ref{e:CanForm}) is less or equal than $(-\Size_3)$.
Thus, \cite[Theorem 2.1]{LL17} (or, alternatively, Corollary \ref{c:dHaY} and (\ref{e:mun})) implies $W(a,Y) > \Size_2 $.
\end{proof}

The implication converse to Corollary \ref{c:w=2d} does not hold true, as one can see from Section \ref{ss:ex},
where an example of $H_{a,Y}$ with 
$W(a,Y) = \Size_4 > 2 \diam Y $ can be found for the case $N=4$.

In this section we show that a stronger connection exists between the 
family of the parameters $\mu_n$  and the metric geometry of the family $Y$ of the interaction centers.
This connection is given below by Theorems \ref{t:gen}-\ref{t:MaxStr} 
in terms of the $m$-sizes $\Size_m $.

\subsection{Generic and maximally structured cases}
\label{ss:t-s}

Let $N \ge 2$ be fixed.
To parametrize rigorously the family of Hamiltonians $H_{a,Y}$,
let us consider in this section $Y$ as a vector in the space $(\RR^3)^N$ of ordered $N$-tuples $y=(y_j)_{j=1}^{N}$
with the entries $y_j \in \RR^3$. We consider $(\RR^3)^N$ as a linear normed space 
with the $\ell^2$-norm $| y |_2 = (\sum |y_j|^2)^{1/2}$.
Then the ordered collection $Y$ of  centers is identified with 
an element of the family $\AAA \subset (\RR^3)^N$ defined by 
$
\AAA := \{ y \in (\RR^3)^N \ : \ y_j \neq y_j' \ \text{ for } \ j \neq j' \} .
$
We consider $\AAA$ as a metric space with the distance function induced by the norm $| \cdot |_2$.

We denote by $M (a,Y)$ and $\mu_n (a,Y)$, $n=1,\dots,M (a,Y)$, the parameter 
$M$ and $\mu_n$, resp., that are associated with the operator $H_{a,Y}$ 
in the way described in Section \ref{s:Asy} (recall that $M$ is the number of distinct 
parameters $\mu_n$).

Let us consider the set $\AAA_\gen$ that consists of $Y \subset \AAA$ that satisfy the following two properties:
\begin{itemize}
\item[(A1)] $M (a,Y)$ and $\{\mu_n (a,Y) \}_{n=1}^{M(a,Y)}$ do not depend on the `strength' tuple $a \in \CC^N$, 
\item[(A2)] $\{\mu_n (a,Y) \}_{n=1}^{M(a,Y)}$ is a subset of  
\[
\left\{ \frac{m-j}{\Size_m (Y) - \Size_j (Y)} \ : \ 2\le m \le N, \ 0 \le j \le N -1, \ j \neq 1, 
\text{ and } \Size_j < \Size_m  \right\} .
\]
\end{itemize}

The following theorem states that, generically, all $\mu_n $ have the form  
$\frac{m-j}{\Size_m - \Size_j }$.

\begin{thm} \label{t:gen}
There exists a subset of $\AAA_\gen$ that is open and dense in $\AAA$.
\end{thm}

The proof is given in Section \ref{ss:proofs}.

The next theorem describe one of the cases where the connection between the sequences of $\mu_n$ and $\Size_m$ 
is especially simple and, simultaneously,  the asymptotics of resonances is maximally structured in the sense 
that the number $M$ of parameters $\mu_n$  takes its maximal possible value $N-1$.

\begin{thm} \label{t:MaxStr}
Let $a \in \CC^N$ and $Y \in \AAA$ be such that the numbers $\Size_n = \Size_n (Y)$, $n=1,\dots,N$, and 
the polynomials $P_b (\cdot)$, $b \in \RR$, which are 
associated with the representation (\ref{e:CanForm}) for 
$D (\zeta) = (-4 \pi)^N \det \Ga_{a,Y} (\ii \zeta) $ and the convention (\ref{e:Pb}),
have the following properties:
\item[(A3)] For $3 \le n \le N$, $P_{-\Size_n } (\cdot) \not \equiv 0$ and $\deg P_{-\Size_n } = N-n$.
\item[(A4)] The sequence $\{ \Size_m \}_{m=1}^N$ is increasing and strictly convex upward in the sense that
 $\Size_m-\Size_{m-1} > \Size_{m+1} - \Size_m >0 $ for $2 \le m \le N-1$ 
 (if $N=2$, this condition is assumed to be fulfilled automatically; recall also that $\Size_1 := (\Size_2 + \Size_0)/2 = \diam Y$).
 
Then:
\item[(i)] $M (a,Y)=N-1$ and 
$\mu_{N-m} (a,Y) = \dfrac1{\Size_{m+1}  - \Size_{m} } $ for $m =1$, \dots, $N-1$;
\item[(ii)] $r_{N-1} = 2$, $r_{n} =1$ for $1 \le n \le N-2$, and, after possible exclusion of a finite number of resonances,  
the multiset $\Si (H_{a,Y})$ can be decomposed into the $2N$ asymptotic sequences (\ref{e:k njt});
\item[(iii)] there exists $R \ge 0$ such that every resonance $k$ with $|k| \ge R$ is simple (i.e., of multiplicity 1).
\end{thm}

Examples in Section \ref{ss:ex} illustrate how large can be the family of sets $Y$ having the property that 
$H_{a,Y}$ satisfies 
(A3)-(A4) for every $a \in \CC^N$.

\subsection{Proofs of Theorems \ref{t:gen} and \ref{t:MaxStr}}
\label{ss:proofs}

Let us introduce the points $\wt T_m := (-\Size_m)+\ii (N-m)$ for $m=0$ and for $2 \le m \le N$.
For $m=0$ and for $2 \le m \le N$, the point $\wt T_m $ is not an endpoint of 
any of the segments $\L_n$, $n=1,\dots,M$ (and so disappears from the computation of the parameters $\mu_n$)
if and only if exactly one of the following conditions hold:
\begin{itemize}
\item[(C1)] $\wt T_m \not \in \{  T_{\beta_j} \}_{j=0}^{\nu}$.
\item[(C2)] $\wt T_m \in \{  T_{\beta_j} \}_{j=0}^{\nu}$, 
$\wt T_m \in \L_n$ for a certain $n$, $1 \le n \le M$, and $\wt T_m$ is not an endpoint of $\L_n$.
\item[(C3)] $\wt T_m \in \{  T_{\beta_j} \}_{j=0}^{\nu}$ and 
$\wt T_m$ lies strictly below a certain segment $\L_n$, $1 \le n \le M$, of the distribution diagram 
(i.e., $(-\Size_m) \in \re [\L_n]$ and $(\wt T_m + \ii \RR_+) \cap \L_n \neq \emptyset$).
\end{itemize}

The main idea of the proofs of Theorems \ref{t:gen} and \ref{t:MaxStr} is that  some type of control 
over the cases (C1)-(C3)  is needed.
While for the proof of Theorem \ref{t:gen} we show that generically (C1) does not hold true,
the assumptions of Theorem \ref{t:MaxStr} exclude the possibility of each of (C1)-(C3).

\emph{Let us prove Theorem \ref{t:gen}.} 
Consider the family $\Aa_{\gen}$ consisting of all $Y \in \AAA$ such that, for a certain $a \in \CC^N$,
\begin{itemize} 
\item[(A5)] each of the points $\wt T_m $ (where $0 \le  m \le N$ and $m \neq 1$) belongs to the set 
$\{ T_{\beta_j} \}_{j=0}^\nu$ associated with $H_{a,Y}$.
\end{itemize}
It is easy to note from the form (\ref{e:Ga}) of $\det \Ga_{a,Y} (\cdot)$ that if $Y \in \Aa_\gen$,
then (A5) is valid for every $a \in \CC^N$. Indeed, (A5) implies  that all $\Size_m$, $2 \le m \le N$, are distinct and that there is no cancellation of
 the leading coefficients of the polynomials $p^{[\si]} (\zeta)$ in (\ref{e:Dterms}) with the degree $N-m$ and $\si$ such that $\al(\si) = - \Size_m$.
On the other hand, these leading coefficients do not depend on $a$.

\begin{prop} \label{p:wtT=T}
There exists a subset of $\Aa_\gen$ that is open and dense in $\AAA$.
\end{prop}
\begin{proof}
It is easy to see that the family $\Aa_1$ introduced in \cite[Section 3.2]{AK18} is a subset of $\Aa_\gen$.
So the proposition is an immediate corollary of \cite[Lemma 3.4]{AK18}.
\end{proof}

\begin{lem} \label{l:AinF}
$\Aa_\gen \subset \AAA_\gen$
\end{lem}
\begin{proof}
Let $Y \in \Aa_\gen$.
Then the validity of conditions (A1)-(A2) of Section \ref{ss:t-s} follows from the definitions of $M$ and $\mu_n$ (see Section \ref{ss:L}) 
and the following fact: if $Y \in \Aa_\gen$, then the set of points $T_{n,h}$, $n=1,\dots,M$, 
$h=1,\dots,\nu_n$, is a subset of the family of points 
$\wt T_m $, where $0 \le  m \le N$ and $m \neq 1$. Indeed, among all the points $T_{\beta_j}$ 
with $\im T_{\beta_j} = N-m$, the point $\wt T_m $ has minimal possible real part.
Since 
\begin{align} \label{e:wtT=T=T}
\wt T_0 = T_0 = T_{M,\nu_M}
\end{align} 
is the only point of the distribution diagram $\L$ 
on the line $\im z = N$, the convexity (L3) of the distribution diagram $\L$ (see Section \ref{ss:L}) implies the desired statement.
\end{proof}

Combining Proposition \ref{p:wtT=T} and Lemma \ref{l:AinF}, \emph{we conclude the proof of Theorem \ref{t:gen}}.

\emph{Let us prove Theorem \ref{t:MaxStr}.}  
It follows from (A3) that the cancellation (C1) does not happen for $m \ge 3$.
For $m=2$, (C1) does not happen because of Lemma \ref{l:Pi2Pi3}. 
For $m=0$, (C1) does not happen because (\ref{e:wtT=T=T})
holds for all $H_{a,Y}$. Thus, (A3) implies (A5). 

It follows from the proof of Lemma \ref{l:AinF} that 
only the points $\wt T_m $ participate in the computation of $M$ and of the parameters $\mu_n$, $n =1 \dots,M$.
The assumption (A4) ensures that all the parameters $\wt T_m $, 
where $0 \le m \le N$ and $m \neq 1$, do participate in this computation.
Thus, $M = N-1$. Now,  $\mu_{N-m} (a,Y) = \frac1{\Size_{m+1}  - \Size_{m} } $ 
and the other conclusions of statement (ii) follow directly from (\ref{e:mun}) and Theorem \ref{t:ra}.

To prove statement (iii) of Theorem \ref{t:MaxStr} it is enough to apply Theorem \ref{t:ra} and to notice that $\om_{M,1} \neq \om_{M,2}$ 
(this follows from Lemma \ref{l:Pi2Pi3} (i) and the procedure for the computation of $\om_{n,j}$ described in Section \ref{ss:ra}).
This completes the proof of Theorem \ref{t:MaxStr}.

\subsection{Examples and the thickness of the case (A3)-(A4).}
\label{ss:ex} 


Let $N=3$. In this case it is easy to give a complete classification of all possible structures of asymptotics of 
$\Si (H_{a,Y})$ due to the fact that $\AAA_\gen = \AAA$. That is, the distribution diagram $\L$ as well as the parameters $\mu_n$, and $\om_{n,j}$ of the sequences 
(\ref{e:k njt}) do not depend on the `strength' tuple $a$. They also satisfy (A2). Note also that 
for $N=3$ the set of tuples 
$Y$ satisfying (A3)-(A4) is generic in the sense that it is open and dense in $(\RR^{3})^3$.
In more details, there are 3 sub-cases of mutual positions of 3 centers 
$\{ y_j\}_{j=1}^3$ composing $Y$, which are described by the corresponding lengths  $\ell_{j,m} = |y_j - y_m|$:
\begin{itemize}
\item[Case 1.] $Y$ consists of the vertices of an  equilateral triangle, i.e., $\ell_{1,2} = \ell_{2,3} = \ell_{3,1}=\diam Y$. 
The asymptotic sequences in this case were found in \cite{AH84} (formally, under the restriction 
$a_1=a_2=a_3 \in \RR$ on the tuple $a$). 
Condition (A3) is satisfied, but the strict convexity part of (A4) is not; $M=1$, $\mu_1 = 1/\diam Y$, and 
$r_1 = 3$.
\item[Case 2.] The points $y_1$, $y_2$, and $y_3$ of $Y$ are on one line. 
Let us number them such that $\ell_{1,2} + \ell_{2,3} = \ell_{1,3} $.
Then $\Size_3=\Size_2$ and condition (A3) fails; $M=1$, $\mu_1 = 1/\ell_{1,3} $, and $r_1 =2$.
\item[Case 3.] The generic case, where $y_1$, $y_2$, and $y_3$ do not fit in the Cases 1 and 2. 
 Let us number them such that $\ell_{1,2} + \ell_{2,3} > \ell_{1,3} = \diam Y$. Conditions (A3)-(A4) are satisfied;
$M=2$; $\mu_1 = 1/(\ell_{1,2} + \ell_{2,3} - \ell_{1,3})$ with the multiplicity $r_1 =1$; 
$\mu_2 = 1/\ell_{1,3} $ with the multiplicity $r_2 =2$.
\end{itemize}

The case of Theorem \ref{t:MaxStr} is not generic when $N=4$. To show this in Proposition \ref{p:nonGen}, we 
start from an example of $H_{a,Y}$ where both of (A3) and (A4) fail.

\begin{ex} \label{ex:nW}
Let us fix an arbitrary $a \in \CC^4$. Consider $Y=\{ y_j \}_{j=1}^4$ with the properties that it is symmetric 
w.r.t. the origin, lies on one line passing through the origin, and satisfy
$|y_1|=|y_4|=c_1$ and $|y_2|=|y_3|=c_2<c_1$. Then the exp-monomials of (\ref{e:Dterms})
that have the lowest possible frequency $(-\Size_4)= - 4 c_1 - 4 c_2 $ do not cancel each other. Thus,
 the following statements are valid for this example:
\begin{itemize}
\item[(i)] $\Si (H_{a,Y})$ has the asymptotics of Weyl type (i.e., $W(a,Y) = \Size_N$);
\item[(ii)] each of the statements (A3), (A4), (A5) does not hold (since $\Size_3 = \Size_2$).
\end{itemize}
Note that, $Y \in \AAA_\gen$.
\end{ex}

\begin{prop} \label{p:nonGen}
Let $a \in \CC^4$. Let $Y$ be as in Example \ref{ex:nW}. 
Then there exists a neighborhood  $\BB_{\vep} (Y)$ of $Y$ such that for all 
$\wt Y \in \BB_{\vep} (Y)$ at least one of the conditions (A3)-(A4) fails for $H_{a,\wt Y}$.
\end{prop}

\begin{proof}
If the frequency $(-\Size_4)$ cancels, then (A3) fails. 
Assume now that $\wt Y \in \BB_{\vep} (Y)$ and (A3) holds for $H_{a,\wt Y}$.
Then it is easy to see that for small enough $\vep$, (C3) is valid and (A4) fails.
\end{proof}

\section{High-energy asymptotics of `physical resonances'}
\label{s:Physics}

The existence of asymptotic structures 
means that the main quantities that characterize the high-energy asymptotics of $\Si (H_{a,Y})$  
from the point of view of Physics are:
the minimum $\mu^{\min}$ of the support of the measure $d \d^{\log} (\cdot)$ 
associated with the asymptotic density function $\d^{\log} (\cdot)$
and the height of the jump at this point $\d^{\log} (\mu^{\min}+0) - \d^{\log} (\mu^{\min}-0)$. 
Indeed, in scattering experiments usually only `sufficiently narrow' and `threshold' resonances $k$ are detected, 
and for this reason they are mostly considered to be the `physically relevant' resonances \cite{E84}. 
Such resonances have small (or, for some models, even zero) value of \emph{resonance width} $|\im k^2 |$
(see \cite{RSIV78}).

In the case of the point interaction Hamiltonian $H_{a,Y}$, with the growth of the energy $\re k^2$,
the resonance width  growth faster for sequences that have larger leading parameter $\mu_n$.
Thus, only the asymptotic sequences corresponding to $ \mu^{\min} = \mu_M = 1 / \diam Y$  
have chances to be detected with reasonable confidence. It is natural to call the density 
\begin{align}  \label{e:anar}
\d^{\log} (\mu_M+0) - \d^{\log} (\mu_M-0) = \frac{r_M  \diam Y}{\pi } 
\end{align} 
of the resonances in the logarithmic strip $\ii \V (\mu_M,w_M)$ (see (\ref{Vmuw}))
\emph{the asymptotic density of narrow resonances}.

\begin{rem}
The study of the density  (\ref{e:anar}) of `more narrow' resonances 
has some similarity with one of the questions considered in 
\cite{HL94,SZ95,SZ99}, where in the context of scattering by a strictly convex $C^\infty$-obstacle, 
a semi-cubic resonance free region 
$\Fr = \{  - C_0 |\re z|^{1/3} + C_1 \le \im z \le 0 \}$  was described \cite{HL94,SZ95}
 and the asymptotics of the counting function in a certain cubic strip 
 $\{  - C_2 |\re z|^{1/3} - C_3 \le \im z \le - C_0 |\re z|^{1/3} + C_1 \}$ (with $C_2>C_0>0$) 
adjacent to $\Fr$ was found \cite{SZ99} under an additional pinching condition on curvatures.
Note that in the case of point interactions the information of the structure of $\Si (H_{a,Y})$ 
given in Theorem 
\ref{t:ra} is more detailed than that of \cite{SZ99} in the case of obstacle scattering.
In particular, Theorem \ref{t:ra}  says that the decay rates $(- \im k_{M,j,t})$ for different 
indices $j$ of the sequences associated with the smallest leading parameter $\mu_M=\mu^{\min}$ are asymptotically equivalent as $|t| \to \infty$. 
In other words, the vertical cross sections of $\ii \V (\mu_M,w_M)$, which 
is the closest to $\RR$ logarithmic strip containing infinite number of resonances, 
remain uniformly bounded.
\end{rem}

Summarizing, we see that, for $H_{a,Y}$, the diameter $\diam Y$ and  the multiplicity parameter $r_M \in \NN$, which take values between 
$2$ and $N$, are the most 
important characteristics of high-energy asymptotics of `physical resonances'.
The following lemma implies that $r_M$ does not depend on the `strength' tuple $a \in \CC^N$.

\begin{lem} \label{l:rM}
The multiplicity parameter $r_M$ is equal to the  maximal integer $m \in [2 ,N]$ such that $\Size_m = m \diam Y$ 
and $\wt T_m \in \L$.
\end{lem}

\begin{proof}
The lemma follows immediately from Sections \ref{ss:L}-\ref{ss:ra} and the definition of the 
points $\wt T_m$ in Section \ref{ss:proofs}.
\end{proof}

\begin{defn}
Let us denote by $r^\narrow (Y)$ the value of $r_M$ for $H_{a,Y}$, and by $\AAA_2$ the set of $Y \in \AAA$
such that  $r^\narrow (Y) = 2$ (i.e., such that $r^\narrow (Y) $ is minimal possible;
we use here the notation of Section \ref{ss:t-s}).
\end{defn}

The following theorem states that, generically, $r^\narrow (Y)=2$.

\begin{thm} \label{t:F2}
There exists a subset of $\AAA_2$ that is open and dense in $\AAA$. 
\end{thm}

\begin{proof}
Let us denote by $\Aa_2$ the set of $Y \in \AAA$ such that $\al (\si_1) \neq \al (\si_2)$ holds for any two different transpositions 
$\si_1$, $\si_2 \in S_{N,2}$ (recall that $\al (\si)$ was defined by \ref{e:al si}). Then it follows from \cite[Lemma 3.4]{AK18} that $\Aa_2$ is open and dense in $\AAA$.
On the other hand, it is obvious from Lemma \ref{l:rM} that $\Aa_2 \subset \AAA_2$.
\end{proof}

So, $r_M >2$ means that the asymptotic density of narrow resonances is atypically large.
Let us consider in more details how this can happen. First, note that for $N \ge 3$,
\begin{align} \label{e:s3=3d}
\text{ $\Size_3 = 3 \diam Y$ \quad  implies \quad $r_M \ge 3$. }
\end{align}
Indeed, $\Size_3 = 3 \diam Y$ means exactly that three of the centers $y_j$ form an equilateral triangle with the side of the length $\diam Y$.
It follows from Lemma \ref{l:Pi2Pi3} (ii), the properties (L3)-(L4) of the distribution diagram $\L$, and Proposition \ref{p:muM}
that $\wt T_3 \in \L$, and so  there is no need to check this condition in Lemma \ref{l:rM} (for $m=3$).

By the following example we show that for $N \ge 4$ the condition 
$\Size_4 = 4 \diam Y$ (i.e., $Y$ contains four distinct centers $y_{j_n}$, $n=1, \dots, 4$,
such that $\diam Y = \ell_{j_1,j_2} = \ell_{j_3,j_4}$)
does not imply $r_M \ge 4$, and, moreover, it does not imply $r_M \ge 3$.
Here and below, we shall set $\ell_{i,j} = |y_{i} - y_{j}|$.

\begin{ex} \label{ex:Y4}
Let $ \wt Y_4 = \{ y_j \}_{j=1}^4$ be the tuple of distinct $\RR^3$ points 
such that 
\[
\ell_{1,2} = \ell_{2,3}=\ell_{3,4}=\ell_{4,1} = \diam \wt Y_4 > \ell_{2,4} \ge \ell_{1,3}
\] 
and let $a \in \CC^4$ be arbitrary.
Then $\Size_4 (\wt Y_4) = 4 \diam \wt Y_4$, but the frequency $\Size_4 (\wt Y_4) $ cancels in the process of summation of 
the exp-monomials (\ref{e:Dterms}). This can be seen by a 
simple modification of arguments of \cite[Section 5.2]{LL17}
or by direct calculations. Since $\Size_3 < 3 \diam \wt Y_4$, we see that $M (a,\wt Y_4) = 2$ and $r^\narrow (\wt Y_4) =2$.
\end{ex}

Let us introduce one more condition on $Y$ with $N \ge 4$:
\begin{itemize}
\item[(A6)] there exist four distinct centers 
$y_{j_n} \in Y$, $n=1, \dots, 4$, such that $\diam Y = \ell_{j_1,j_2} = \ell_{j_3,j_4}$
and $\{y_{j_i}\}_{i=1}^4$ cannot be reordered to satisfy the condition of Example \ref{ex:Y4}.
\end{itemize}

\begin{thm} \label{t:rM>2}
Let $N \ge 4$. Assume that condition (A6) holds true or $\Size_3 (Y) = 3 \diam Y$.
Then $r^\narrow (Y) > 2 $.
\end{thm}

\begin{proof}
The case $\Size_3 = 3 \diam Y$ follows from (\ref{e:s3=3d}).
Assume now that $\Size_3  \neq  3 \diam Y$ and that (A6) holds. Note that (A6) implies $\Size_4 = 4 \diam Y$.
Therefore, due to Lemma \ref{l:rM}, it is enough to prove that $(-\Size_4) $ does not cancel
in the process of summation of the exp-monomials (\ref{e:Dterms}) corresponding to permutations $\si\in S_{n,4}$.

Consider the family $B_1$ of unordered sets $Y_4$ consisting of 4 points of $Y$ such that a certain ordering 
$\wt Y_4 := \{y_{m_i}\}_{i=1}^4$ of $Y_4$ satisfies the conditions of Example \ref{ex:Y4}.
To such $\wt Y_4 $, we put into correspondence the two families, $E_1 (\wt Y_4)$ and $E_2 (\wt Y_4)$,
each of which is an unordered pair of unordered pairs of centers: $E_1 (\wt Y_4) = \{ \{y_{m_1},y_{m_2}\}, \{y_{m_3},y_{m_4}\}\}$ and 
$E_2 (\wt Y_4)= \{ \{y_{m_2},y_{m_3}\}, \{y_{m_4},y_{m_1}\}\}$.
The assumption $\Size_3  \neq  3 \diam Y$ implies that for each $Y_4 \in B_1$ there exists a unique
unordered pair $E (Y_4)$  
of the form $ \{E_1 (\wt Y_4),E_2 (\wt Y_4)\}$ 
i.e., $\Size_3  \neq  3 \diam Y$ implies that various reorderings of $\wt Y_4$ 
that preserve the conditions of Example \ref{ex:Y4} produce the same unordered pair 
$E (Y_4)$ (of pairs of pairs of centers). 
(For the case where $\Size_3  =  3 \diam Y$ and $\{y_{m_i}\}_{i=1}^4$ 
are vertices of a regular tetrahedron, the above statement obviously fails.)

The family of all  pairs of pairs of centers produced by 
the maps $E_i (\cdot)$, $i=1,2$, shall be denoted by $B_2$, i.e., $B_2 = E_1 [B_1] \cup E_2 [B_2] $.

The class $\Sf_1$ of all permutations $\si \in S_{N,4}$ that have the sign $\ep_{\si} = -1$ 
and produce  exp-monomials (\ref{e:Dterms}) with the frequency $(-\Size_4)$ 
consists of permutations $[m_1 m_2 m_3 m_4]$ and $[m_4 m_3 m_2 m_1]$ 
associated with 
a certain ordered tuple $\wt Y_4$ having the property that its unordered version $Y_4$ belongs to 
$B_1$. These permutations produce the exp-monomials 
$\frac{-e^{-\Size_4} }{(\diam Y)^4} \prod_{j=\si (j)} (-\zeta -A_j) $. To each pair 
$[m_1 m_2 m_3 m_4]$ and $[m_4 m_3 m_2 m_1]$ of such permutations 
we put into correspondence the two permutations 
$[m_1 m_2][m_3 m_4]$ and $[m_2 m_3][m_4 m_1]$ associated with $E_1 (\wt Y_4)$ and $E_2 (\wt Y_4)$, resp., 
and we unite all permutations generated in this way by $E_i (\wt Y_4)$, $i=1,2$, into the class $\Sf_2$.
The permutations $[m_1 m_2][m_3 m_4]$ and $[m_2 m_3][m_4 m_1]$ produce the exp-monomials (\ref{e:Dterms})
of the form $\frac{e^{-\Size_4} }{(\diam Y)^4} \prod_{j=\si (j)} (-\zeta -A_j) $.
So after summation of all exp-monomials generated by permutations of the classes $\Sf_1$ and $\Sf_2$ the 
coefficients of the highest order cancel and the associated polynomial has a degree $< N-4$.

Assumption (A6) means that $\{ \{y_{j_1},y_{j_2}\},\{y_{j_3},y_{j_4}\}\} $ does not belong to $B_2$.
However, the exp-monomial produced by the associated permutation $\si =[j_1 j_2][j_3 j_4]$ has 
the sign $\ep_{\si} = 1$ and the frequency $(-\Size_4)$. 
Hence the frequency $(-\Size_4)$ does not cancel and, 
moreover, $\deg P_{-\Size_4} = N-4$. 
This proves $\wt T_4 \in \L$, and due to Lemma \ref{l:rM}, completes the proof of the theorem.
\end{proof}

Under symmetries of $Y$ we understand isometries $\I$ of $\RR^3$ such that $\I [Y]=Y$.
Moreover, two symmetries $\I_1$ and $\I_2$ will be considered identical (equivalent) if 
$\I_1 (y_j) = \I_2 (y_j)$ for all $y_j \in Y$.
So the group $\Gr$ of symmetries of $Y$ consists of the defined above classes of equivalence of 
isometries of $\RR^3$ that map $Y$ onto $Y$.

\begin{cor} \label{c:Gr}
Let $\Gr$ be the group of symmetries of $Y = \{ y_j \}_{j=1}^N$.
Assume that there exist $\I \in \Gr$ and $y_{j_1}, y_{j_2} \in Y$ 
with the following properties:
\item[(i)] $|y_{j_1} - y_{j_2}| = \diam Y$, 
\item[(ii)] the sets $\{y_{j_1},y_{j_2}\}$ and $\{\I (y_{j_1}) , \I (y_{j_2})\}$ are disjoint, 
\item[(iii)] the 4-tuple $\{y_{j_1},y_{j_2}, \I (y_{j_1}) , \I (y_{j_2}) \}$ cannot be reordered in such a way that it satisfies 
the conditions of Example \ref{ex:Y4}.

Then $r^\narrow (Y) >2$.
\end{cor}
\begin{proof}
The statement follows easily from Theorem \ref{t:rM>2}.
\end{proof}

Theorem \ref{t:rM>2} and Corollary \ref{c:Gr} suggest that $r^\narrow (Y) >2$ 
whenever $Y$ possesses a rich enough group of symmetries $\Gr$. 
This claim can be supported by the following examples.

\begin{ex} \label{ex:sym}
Let the tuple $a$ be arbitrary. 
Then $r^\narrow (Y) >2$ in each of the following cases:
\item[(i)] $Y$ is the set of all vertices of a Platonic solid;
\item[(ii)] $Y$ is the set of all vertices of a right prism.

These statements can be shown by using Theorem \ref{t:rM>2} or by the direct calculation of the exp-monomials
with the frequencies $(-m)\diam Y$, $m \in \NN \cap [3,N]$.
\end{ex}

The $N-1$ possible values $2,3,\dots,N$ of $r^\narrow$ naturally 
split the family $\AAA$ of all possible tuples $Y$ of interaction centers into $N-1$ classes.
In our opinion, the interplay between these classes and the groups of symmetries of $Y$ deserves 
an additional study.

\section{Remarks on resonances of quantum graphs
}
\label{s:G}

To study the asymptotic structure of $\Si (H)$ for quantum graph Hamiltonians, 
the above approach requires some adaptation that is done in the next subsection.
Subsection \ref{ss:PhC} briefly addresses the resonances of 1-D photonic crystals considering them 
as a simple case of `weighted' quantum graphs.

\subsection{Structure of resonance asymptotics for quantum graphs} 

\label{ss:QG}

The asymptotics $\N (R) = \frac{C R}{\pi}+O(1)$ of the resonance counting function $\N (R) $ 
for a Hamiltonian $H$ associated with a noncompact quantum graph $\G$ consisting of a finite 
number of edges have been derived in \cite{DEL10,DP11}, where, 
in particular, Weyl and non-Weyl types of asymptotics for $\G$ were introduced  depending on 
the value of the constant $C \ge 0$.
The characterizations of these two cases of asymptotics were obtained in terms of types of couplings at the vertices. 
A number of examples were elaborated (see \cite{L16,EL17} and the references therein).

We show that, at least under additional conditions, the asymptotic structures which are somewhat similar, 
but more complex, than those of $\Si (H_{a,Y})$,
 exist in the multiset of resonances $\Si (H)$ of a quantum graph. This leads to another versions
 of the semi-logarithmic strip counting 
 \begin{align*}
& \N (\mu,\ga,R)  := \# \{ k \in \Si (H) :   
 - \mu \ln (|\re k|+1) - \ga \le \im k  , \  |k|\le R \} , \quad \mu, \ga \in \RR, \\
& \N (\mu,+\infty,R)  := \N_H (R) \text{ for } \mu \ge  0.
 \end{align*}
and the asymptotic density function 
\begin{align} \label{e:AsDmuga}
\d (\mu,\ga) := \lim_{R\to \infty} \frac{ \N (\mu,\ga,R)}{R} 
\end{align}
depending now on the two parameters $\mu \ge 0$ and $\ga \in (-\infty,+\infty]$.
(For $\mu <0$ and $\ga \neq +\infty$, one can put $\d (\mu,\ga) := 0 $ due to 
the fact that $\# \{ k \in \Si (H) \cap \CC_+ \} < \infty$, see Sections \ref{sss:Kir}-\ref{sss:KSH}).

With such a definition the effective size of $H$ introduced in \cite{DEL10,DP11} is equal to 
$ \frac{\pi \d (0,+\infty)}{2} $
and  \cite[Theorem 3.2]{DP11} and \cite[Theorem 3.3]{DEL10}  can be rewritten as 
the inequality 
\[
\d (0,+\infty) = \d (\mu,+\infty) \le \frac{2}{\pi} \sum_{j \in J} \rho_j, 
\]
where $\rho_j$ are the lengths of the internal 
(i.e., bounded) edges of the graph, $J$ is the finite set  that indexes all such edges, and $\mu \ge 0$ is arbitrary.

\subsubsection{The case of Kirchhoff's boundary condition at vertices}
\label{sss:Kir}

In this subsection we follow the settings of \cite{DP11} and consider on a non-compact quantum graph $\G$ 
the self-adjoint Hamiltonian $H_\G$ associated with the differential expression $ - \dd^2/\dd x^2$ 
with the continuity condition and the Kirchhoff (boundary) condition at each of the vertices.

The set of resonances $\Si (H_\G)$ is the set of $k \neq 0$ such that 
there exists a resonant mode $f(\cdot)$ (continuous and $L^2_{\CC,\loc}$ on $\G$) 
satisfying $f'' = - k^2 f$ on $\G$,
the Kirchhoff conditions at each vertex, 
and the classical radiation condition on each lead (i.e., semi-infinite edge). 
The Kirchhoff condition means that the sum of outgoing derivatives of $f$ at a vertex is equal to $0$.
The radiation condition after the identification of a lead with $[0,+\infty)$ takes the form $f(x) = f(0) \ee^{\ii k x}$. 

To make $\Si (H_\G)$ a multiset, each resonance $k$ is equipped with a multiplicity, 
which is multiplicity of $k$ as a zero of a specially constructed analytic function $F (z) = \det A (z)$,
where $A(z)$ is a matrix produced by a plugging of the exponential fundamental system of solutions $e^{\pm \ii k x}$ 
on each edge into the  continuity, Kirchhoff's, and radiation conditions (see \cite[Theorem 3.1]{DP11}  for  details).

\begin{rem}
In the process of construction in \cite{DP11} of the matrix-valued function $A(\cdot)$, 
the Kirchhoff conditions are divided by $\ii k$, which is permissible because $k=0$ is excluded from the consideration.
Let us observe that, on one hand, $k = 0$ always satisfies the resonance conditions, and on the other hand after 
the transition to the settings with the energy-type spectral parameter $\la = k^2 $, the point $\la =0$
is not a pole, but a branching point of the generalized extension of the resolvent $(H_\G - \la)^{-1}$.
The exclusion of $k=0$ is usual for some types of 1-D resonance problems 
(see e.g. \cite{CZ95,Ka14,KLV17})
and makes the above definition of resonances slightly different from that of \cite{DZ17,Z17},
where $0$ is permitted to be a resonance of a finite multiplicity. Obviously, this difference does not influence 
the results of \cite{DP11,DEL10} and of the present paper on the asymptotics at $\infty$; however it influences 
the formulation of Theorem \ref{t:Kir} below.
\end{rem}

It is shown in \cite[Theorem 3.2]{DP11} that $\Si (H_\G) $ lies in a certain horizontal strip 
$\{ z \in \CC :  - \ga \le \im z \le 0  \}$ (formally, the theorem states this for a certain strip 
$\{ |\im z| \le \ga \}$, however $H_\G = H_\G^* \ge 0$, and so $\Si (H_\G) \cap \CC_+ = \varnothing$). 
Hence, the measure $d \d^{\log} $ does not carry information on the internal structure of $\Si (H_\G)$.
Indeed, it consists of one mass $\d^{\log} (+\infty)$ at the point $\mu=0$ and so gives only the information about 
the total asymptotic density  $\d^{\log} (+\infty)$.

To capture the structure of the asymptotics of $\Si (H_\G)$, let us introduce another asymptotic density function 
\begin{align} \label{e:AdHor}
\d^\hor (\ga) := \lim_{R\to \infty} \frac{\N (0,\ga,R) }{R}  , \quad  \ga \in (-\infty,+\infty] .
\end{align}
To show that $\d^\hor $ is an adequate tool, we consider the case 
where the  lengths of all internal edges $\rho_j$ are commensurable (in the sense of 
(\ref{e:com}))
 and make the observation that for this case the measure $d \d^\hor $ generated by the monotone function $\d^\hor$ consists of a finite number of point masses.

\begin{thm} \label{t:Kir}
Assume that $\Si (H_\G) \neq \varnothing$ and 
\begin{align} \label{e:com}
\text{the lengths $\rho_{j_1}$ and $\rho_{j_2}$ of two arbitrary internal edges of $\G$
satisfy $\rho_{j_1}/\rho_{j_2} \in \QQ$.}
\end{align} 
Then there exist $\g \in \NN $, $\beta \in (0,+\infty)$, 
and infinite sequences $ \{k_{n,t}\}_{t \in \ZZ} \subset \CC$, $n = 1, \dots, \g$, that 
satisfy 
\begin{align} 
\Si (H_{\G})  =   \bigcup_{n=1}^{\g} \{k_{n,t}\}_{t \in \ZZ} \setminus \{0 \} \quad \text{ and } \quad
\beta k_{n,t}  =  2 \pi t - \ii \Ln |\xi_n| + \Arg_0 \xi_n ,
\label{e:k nt}
\end{align}
where $\xi_n$, $n=1, \dots, \g$, are certain complex algebraic numbers satisfying $|\xi_n| \ge 1$.
\end{thm}

\begin{proof} 
We can assume that the index set $J$ is of the form  $\{ 1, \dots, N \}$.
The function $F (z) = \det A(z)$ that produces the set of resonances as the set
of its zeroes $k \neq 0$  is an exponential polynomial obtained by the summation of exp-monomials of the form 
$(-1)^{\tau_m} K_m \exp \left( \ii z  \sum_{j \in J_m} (-1)^{\wt \tau_{m,j} } \rho_j\right)$ (see \cite[Theorem 3.1]{DP11}), 
where the index $m$ passes through a certain finite index set $\M$, the sequences $\{\tau_m \}_{m \in \M}$,
$\{\wt \tau_{m,j} \}_{m \in \M}^{j \in J_m}$, $\{ K_m \}_{m \in \M}$ are subsets of $\ZZ$, and, for each $m \in \M$, 
$J_m$ is a subset of $J$. 
Writing $F $ in the canonical form we get 
\begin{align} \label{e:CanFormG}
F(z) = \sum_{l=0}^\nu C_l  e^{\ii b_l z } ,
\end{align} 
where $\{ C_l \}_{l=0}^\nu \subset \ZZ \setminus \{0\}$ and all
$b_l$, $l=0, \dots,\nu$, have the form $\sum_{j \in \wt J_l} (-1)^{\tau_{l,j}} \rho_j $ with $\wt J_l \subset J$ 
and $\tau _{l,j} \in \ZZ$. We can also assume that the sequence of $b_l$ is strictly increasing.

Supposing now that (\ref{e:com}) holds, one can see that the arguments similar to that of \cite[Section 12.4]{BC63} 
are applicable to the function $\wt F (z) :=e^{- \ii b_0 z} F(z)$. Indeed,
let $\beta >0 $ be such that $\rho_j = \beta d_j$ with $d_j \in \NN$ for all $j \in J$. Then the numbers $\wt b_l := b_l - b_0$
are also commensurable and $\wt b_l = \beta \wt d_l$, $l=0,\dots,\nu$, with $\wt d_l \in \NN$. 
Now one can easily obtain the statement of the theorem writing $F(z)=0$ as 
$P (e^{\ii \beta z})=0$, where $P (\cdot)$ is a polynomial with integer coefficients.
\end{proof}

Thus, in the commensurable case of Theorem \ref{t:Kir}, 
the asymptotic density $\d^\hor (\cdot)$ is a piecewise constant function with a finite number of jumps 
at the points $\beta^{-1} \Ln |\xi_n|$. 
Depending on the level of noise, the asymptotic sequences with smaller $\Ln |\xi_n|$ have more chances to be detected 
in scattering experiments because they produce narrower resonances. 

It is natural to call the resonances $k$ lying on 
$\RR$ \emph{embedded resonances} because the corresponding points $\la =k^2$ in 
the `energy' plane are embedded into 
the continuous spectrum $\si_\cont (H_\G) = [0,+\infty)$ of $H_\G$. By \cite[Theorem 2.3]{DP11}, 
every embedded resonance $k$ corresponds 
to an eigenvalue $k^2$ of $H_\G$ embedded into $\si_\cont (H_\G) $.
The case $c \in \QQ \cup [0,1]$ of the example of \cite[Section 6]{DP11} gives $H_\G$ 
with infinite number of embedded resonances (in more general settings, 
equispaced sequences of embedded resonances have been considered in \cite{EL10}).

The following general statement describes  in the commensurable case 
the situation when $\d^\hor (\cdot)$ has a jump at $0$.

\begin{cor}
Suppose that (\ref{e:com}) holds. Then the following statements are equivalent:
\item[(i)] $H_\G$ has an infinite number of embedded eigenvalues (and so, of embedded resonances);
\item[(ii)]  $H_\G$ has at least one nonzero eigenvalue or at least one embedded resonance;
\item[(iii)] $\min_{1 \le n \le \g} |\xi_n| = 1$ in the settings of Theorem \ref{t:Kir}.
\end{cor} 
\begin{proof}
The statement follows immediately from  Theorem \ref{t:Kir}.
\end{proof}


In the non-commensurable case, the structure of $\Si (H_\G)$ and the structure of the measure 
$d \d^\hor $ generated by the 
monotone function $\d^\hor (\cdot)$ deserve an additional study, which is connected with 
the theory of distribution of zeroes of exponential polynomials (see \cite{BG12} and references therein).
The main parameters of   high-energy asymptotics for the general case 
are: 
(i) the total asymptotic density $\d^\hor (\ga^{\max} + 0) = \d^\hor (+\infty)$, which was  considered in 
\cite[Theorem 1.2]{DP11}, 
(ii) the minimum $\ga^{\min}$ and the maximum $\ga^{\max}$ 
of the (closed) support $ \supp (d \d^\hor)$ of the measure $d \d^\hor$, 
(iii) $\# \{ k \in \Si (H_\G) : \im k > -\ga^{\min} \}$, and (iv) $\lim_{\de \to + 0} \frac{\d^\hor (\ga^{\min} + \de) - \d^\hor (\ga^{\min})}{\de}$.

It is possible to strengthen slightly \cite[Theorem 1.2]{DP11} by the following observation.

\begin{cor} \label{c:d>0}
Assume that $\Si (H_\G) \neq \emptyset$. Then $\Si (H_\G)$ consists of an infinite number of resonances 
and their asymptotic density $\d^\hor (+\infty) (= \lim_{R \to +\infty} \frac{\N_{H_\G} (R)}{R})$ is positive.
\end{cor}

\begin{proof}
It follows from $\Si (H_\G) \neq \emptyset$ that $\nu \ge 1$ in (\ref{e:CanFormG}). Then the desired statement can be easily 
obtained from \cite[Theorem 12.5]{BC63} (see also \cite[Theorem 3.2]{DP11} and references in \cite{DEL10,DP11}).
\end{proof}

\subsubsection{The case of general self-adjoint local coupling}
\label{sss:KSH}

In this subsection the local case of a more general  
Kostrykin-Schrader-Harmer 
coupling \cite{KS99,H00} (in short, local KSH-coupling) is considered with the use of the notation and settings 
of \cite{DEL10} (see also \cite{L16} for a detailed exposition and the literature review). 

Let $|\G|$ be the number of vertices of the  quantum graph $\G$ and 
$\{ X_n \}_{n=1}^{|\G|}$ be
the set of the vertices. Let $\deg X_n$ be the degree of  the vertex $X_n$, i.e, the number of edges connected 
to $X_n$.

The quantum graph Hamiltonian $H = H_{\G,U}$ and 
the multiset of its resonances $ \Si (H_{\G,U})$ are 
defined similarly to Section \ref{sss:Kir},
but with the Kirchhoff and continuity conditions replaced at each $X_i$ by the local KSH-coupling.
The latter means that, at each vertex $X_n$, the condition $(U_n-I) \Psi_n + \ii (U_n+I) \Psi'_n = 0$
is satisfied with a certain unitary $\deg X_n \times \deg X_n$ matrix $U_n$, where the vector $\Psi_n$ consists of 
the limits of $f(\cdot)$ at $X_n$ along every edge $\E_j$ connected to $X_n$, and the vector $\Psi'_n$ consists 
of the corresponding limits of outwards derivatives. The unitary matrix $U$ in the notation $ H_{\G,U}$ is composed 
of the diagonal blocks $U_n$, see \cite{DEL10,L16} for details. 

The next theorem states, roughly speaking, that generally the structure of $ \Si (H_{\G,U})$ 
is the combination of the two types considered in Section \ref{sss:Kir} 
(with the Kirchhoff and continuity coupling) and Section \ref{s:Asy} (for point-interactions).

Let $\wt \K (\ga) := \{ k \in \Si (H_{\G,U}) \ : \ \im k \ge -\ga \}$ and let 
$\mu^{\max}$ be the supremum of the support of the measure $d \d^{\log} $ for the Hamiltonian $H_{\G,U}$
($\mu^{\max}$ is taken to be equal to $-\infty$ if $\d^{\log} (\cdot) \equiv 0$ on $\RR$).

\begin{thm} \label{t:KSH} 
(i) If $\mu^{\max} \le 0$, then there exists $\wt \ga \in \RR $ such that $\Si (H_{\G,U}) = \wt K (\wt \ga) $.

\noindent (ii) Assume that $\mu^{\max} > 0$. Then there exists $\wt \ga \in \RR $, $M \in \NN $, 
a strictly decreasing finite sequence $\{ \mu_n \}_{n=1}^M \subset \RR_+$, finite sequences $\{r_n\}_{n=1}^M \subset \NN$, 
$\{ \om_{n,j} \}_{j=1}^{r_n} \subset \CC \setminus \{0\}$, $\{ t_{n,j}^\pm \}_{j=1}^{r_n} \subset \NN$ for $n=1$,\dots ,$M$, and 
infinite sequences $\K_{n,j}^\pm = \{k_{n,j,t}^\pm\}_{t=t_{n,j}^\pm}^{+\infty} \subset \CC$, $n=1$, \dots, $M$, $j=1$, \dots, $r_n$,
 with the following properties:

(ii.a) $ \Si (H_{\G,U}) = \wt K (\wt \ga) \cup \bigcup_{n=1}^M \K_n $ (taking into account multiplicities),
where \linebreak \qquad \qquad $\K_n := (\bigcup_{j=1}^{r_n} \K_{n,j}^-) \,  \cup \, (\bigcup_{j=1}^{r_n} \K_{n,j}^+)$ for $n=1$, \dots , $M$.

(ii.b) Each of $\K_{n,j}^\pm $ has the asymptotics (\ref{e:k njt}) as $t \in \NN$ goes to $+\infty$.

(ii.d) If the set $\wt K (\wt \ga)$ is infinite, then 
\[
\lim_{R\to \infty} \frac{\# \{k \in \wt K (\wt \ga) \ : \ |k| \le R \}}{R} =  \d^\hor (\wt \ga +0) = 
\d^\hor (+\infty)> 0.
\]

\noindent (iii) Let additionally the commensurability condition (\ref{e:com}) hold. Then in each of the cases (i) and (ii) 
the set $\wt \K (\wt \ga)$ is either finite, or satisfies the following property: 
there exist numbers $\beta>0$, $\g \in \NN $, $M_0 \in \NN $, finite sequences 
$\{ \xi_n \}_{n=1}^{\g} \subset \CC \setminus \DD_1 (0)$, 
$\{ \wt t_n \}_{n=1}^{\g} \subset \NN$, $\wt \K_0 := \{ \wt k_t \}_{t=1}^{M_0}$, and 
infinite sequences $ \{\wt k_{n,t}^\pm\}_{t=\wt t_{n}}^{+\infty} \subset \CC$ such that
\begin{align} 
\wt \K (\wt \ga)  & =  \wt \K_0 \cup  \bigcup_{n=1}^{\g} 
\left(\{\wt k_{n,t}^- \}_{t=\wt t_{n,j}}^{+\infty} \cup \{\wt k_{n,t}^+ \}_{t=\wt t_{n,j}}^{+\infty} \right) 
\text{ and } \notag \\
\beta \wt k_{n,t}^\pm & =  \pm 2 \pi t - \ii \Ln |\xi_n| + \Arg_0 \xi_n + o(1) \quad  \text{ as $t \to +\infty$}.
\label{e:kntAs}
\end{align}
\end{thm}

\begin{proof}
A function $F(z)$ such that $k \in \Si (H_{\G,U})$ if and only if $k \neq 0$ and $F(k) = 0$ 
is constructed in \cite{DEL10}.
Moreover, the multiplicities of the resonances are the multiplicities of the corresponding zeros of $F$.
Let us consider a modified version of $F$ given by $\wt F (\zeta) = F(-\ii \zeta)$. Then \cite[Theorem 3.1]{DEL10} 
implies that $\wt F (\zeta)$ is an exponential polynomial of the form (\ref{e:CanForm}) with a 
strictly increasing sequence of frequencies $\wt \beta_j$ that are not necessarily nonpositive.
Now the function $D (\zeta)  = e^{-\wt \beta_\nu} \wt F (\zeta)$ has the same zeros as $\wt F (\zeta) $ and 
has the form (\ref{e:CanForm}) with a strictly increasing sequence of frequencies $\beta_j$ so that $\beta_\nu = 0$.
The construction of the distribution diagram given in Section \ref{ss:L} and in more details in \cite{BC63} applies to $D (\cdot)$.
In particular, in the terminology of \cite{BC63}, the zeros of $D (\cdot)$ lie in a finite number of logarithmic curvilinear
strips and possibly one neutral (horizontal) strip $\{ |\im z | \le \ga_0 \} $  
(this is the statement of \cite[Theorem 3.1]{DEL10}).
The logarithmic strips contain a infinite number of resonances and are necessarily 
retarded because the self-adjoint and lower semi-bounded from below 
operator $H_{\G,U}$ has at most a finite number of resonances in $\CC_+$.
Thus, the arguments of Section \ref{ss:ra} applied to the zeroes of $D $ in logarithmic strips and the arguments of 
\cite[Sections 12.4-6]{BC63} applied to the zeros of $D$ in the neutral strip easily complete 
the proof of the theorem.
\end{proof}

We see now that the two level asymptotic structure of $\Si (H_{\G,U})$ is captured by the asymptotic density 
function $\d (\mu,\ga)$ of (\ref{e:AsDmuga}). The measure $d \d^{\log} (\cdot) = d \d (\cdot,0)$ has point masses at 
the numbers $\mu_n$ corresponding to the logarithmic asymptotic sequences and, in the case 
where the strip $\{ |\im z | \le \wt \ga \} $ contains an infinite number of resonances, it has also a point mass at $0$. 
If the infimum $\mu^{\min} $ of the support of the measure $d \d^{\log} (\cdot)$ is equal to $0$, then 
the measure $d \d (0, \cdot) = d \d^{\hor} (\cdot)$ is responsible for the internal structure of 
$\Si (H_{\G,U})$ in the strip $\{ |\im z | \le \wt \ga \} $ and for high-energy asymptotics 
of narrow resonances.

If $\mu^{\min}>0$, then $\mu^{\min} = \mu_M$. The high-energy asymptotics of `physical resonances' is described 
in this case by the sequences 
 $\K_{M,j}^\pm $, $j=1$, \dots, $r_M$, and, on a more rough level, by the asymptotic density 
$\d^{\log} (\mu_M+0) - \d^{\log} (\mu_M-0) $ in the corresponding logarithmic strip.

\subsection{Resonances of 1-D photonic crystals}

\label{ss:PhC}

A typical 1-D photonic crystal (a multi-layer optical cavity) is described by 
the variable dielectric permittivity $\vep (x) >0$, $x \in \RR$, which is a piecewise constant function on $\RR$ 
with a finite number of steps. That is, there exists a finite partition
$-\infty = x_{-1} < x_0 <  \dots < x_{N} < x_{N+1} = +\infty$ with $N \in \NN$ 
and constants $\ep_j \in \RR_+$ such that 
$\vep (x) = \ep_j $ for all $x \in (x_{j-1},x_j)$ , $n=0$, \dots, $N+1$.
For $1 \le j \le N$, the intervals $(x_{j-1},x_j)$ represents idealized infinite plane layers of a material with 
the permittivity $\ep_j$. The semi-infinite intervals $(-\infty, x_0)$ and $(x_N,+\infty)$ 
represents the homogeneous  outer medium (it is assumed often that the corresponding permittivities $\ep_0$ and $\ep_{N+1}$ are equal,
however this is not important for this section). 

For electromagnetic waves that pass normally to the interfaces of the layers,
the Maxwell system can be reduced to the wave equation for a  nonhomogeneous string 
$ \vep(x) \pa_t^2 v (x,t) = \pa_s^2 v (x,t) $ (see \cite{KLV17} and references therein). The corresponding operator 
$H_\vep:= - \frac{1}{\vep(x)} \pa_x^2$ is self-adjoint and nonnegative in the weighted Hilbert space 
$L_\CC^2 (\RR; \vep(x) dx )$. 

The resonances of $H_\vep$ are the complex numbers $k \neq 0$ such that 
there exists a nontrivial solution $f$ to the equation $f'' = -k^2 \vep(x) f (x)$ satisfying 
radiation conditions on the outer intervals $(-\infty, x_0)$ and $(x_N,+\infty)$.
The latter means that $f(x) = f(x_0) e^{ -\ii k (x-x_0)\ep_0^{1/2}}$ for $x<x_0$ and 
$f(x) = f (x_N) e^{\ii k (x-x_N)\ep_{N+1}^{1/2}}$ for $x>x_N$.
The multiplicity of a resonance $k$ can be defined as the multiplicity of the zero at $k$ of the corresponding 
Keldysh characteristic determinant (see \cite{K13,KLV17}) or, equivalently, 
via the algebraic multiplicity of the eigenvalue $k$ of a special operator \cite{CZ95}. 
The multiplicity of each resonance is finite.

The approach to resonances through the corresponding Keldysh characteristic determinant is essentially equivalent 
to the approach of \cite{DP11} with the determinant of a matrix-valued function $A(z)$ constructed by the
coupling conditions. Indeed, we can consider a 1-D photonic crystal as a \emph{weighted quantum graph} with the simple linear connectivity.
It consists of edges $[x_{j-1},x_j]$, which are internal for $j=1$, \dots, $N$ and external for $j=0$ and $j=N+1$,
equipped with the differential expressions $\frac{1}{\ep_j} \pa_x^2 $ with the constant 
coefficient $1/\ep_j$.
The coupling of the graph is given by the conditions of continuity of $f$ and $f'$, $f (x_j-0) = f(x_j+0)$ and 
$f' (x_j-0) = f'(x_j+0)$, $j=0,\dots,N$. Note that the latter condition is a simple version of Kirchhoff's 
condition for the case when only two edges go out of a vertex $x_j$.
(This weighted quantum graph slightly does not fit into the class of the weighted graphs of 
\cite[Section 8]{DEL10} because its coupling conditions for the derivatives are different.)

\begin{thm}
Assume that the multiset of resonances $\Si (H_\vep)$ of $H_\vep$ is nonempty.
Then $\# \Si (H_\vep) = \infty$ and the following statements hold:
\item[(i)] There exists $\ga_0 >0$ such that $\Si (H_\vep) \subset \{ z \in \CC \ :  - \ga_0 \le \im z <0 \}$ and 
$0< \d^\hor (\ga_0) < +\infty$, where $\d^\hor (\cdot)$ is the asymptotic density function  
defined by (\ref{e:AdHor}) for $H=H_\vep$.
\item[(ii)] Assume, additionally, that $\frac{(x_j - x_{j-1})\ep_{j}^{1/2} }{(x_{j+1} - x_j)\ep_{j+1}^{1/2} } \in \QQ$ for 
$1 \le j \le N-1$. Then there exist numbers $\beta>0$, $r \in \NN $, 
and infinite sequences $ \{k_{n,t}\}_{t \in \ZZ} \subset \CC$, $n = 1, \dots, r$, such that 
\begin{align} 
\Si (H_{\vep})  =   \bigcup_{n=1}^{r} \{k_{n,t}\}_{t \in \ZZ}  \quad \text{ and } \quad
\beta k_{n,t}  =  2 \pi t - \ii \Ln |\xi_n| + \Arg_0 \xi_n , \notag
\end{align}
where $\xi_n$, $n=1, \dots, r$, are certain complex numbers satisfying $|\xi_n| >1$.
\end{thm}

\begin{proof}
To obtain the Keldysh characteristic determinant $F(z)$, the construction of the matrix-valued function $A(z)$ from \cite{DP11} can be applied resulting 
in $F(z) = \det A(z)$. Since in this process
the derivatives $f'(x_j-0)$ are coupled only with the derivatives $f'(x_j+0)$, the multiplicative factors $\ii z$ are eliminated 
from $A(z)$ (similarly to \cite{DP11} and the definition of resonances in \cite{CZ95,K13,KLV17}). 
Hence, $F(z) $ takes the form (\ref{e:CanFormG}) with certain $C_l \in \RR \setminus \{0\}$.
Thus, the arguments of \cite{DP11}, Theorem \ref{t:Kir}, and Corollary \ref{c:d>0} can be applied 
to obtain all the statements of the theorem except the statement that $\Si (H_\vep) \cap \RR = \varnothing$,
which is well-known.
\end{proof}

\section{Discussion on other classes of m-D Hamiltonians}
\label{s:Dis}

We have described two  structural levels of the asymptotic distribution of the set of resonances $\Si (H)$ 
for 3-D Schrödinger Hamiltonians $H=H_{a,Y}$ with point interactions and have shown how 
these results can be adapted to quantum graphs $H_{\G,U}$. 
This allows us to define parameters corresponding to the asymptotics of 
the `physically relevant' part of $\Si (H)$, which consists of resonances lying in the logarithmic or horizontal 
strip 
which situated closest  to the real line.
The method relies on the fact that, for these two classes, resonances are zeros of exponential 
polynomials.

One can notice similarities between some of the features of the asymptotic structures described above and 
the known facts about resonances (or scattering poles) for Dirichlet Laplacians $H_{\OO} = - \De$ arising 
in obstacle scattering. The study of resonances of  $H_{\OO}$ requires more involved analytical tools and 
the presently available results on the asymptotic structure of the set 
$\Si (H_\OO)$ of resonances of $H_\OO$ are obtained 
under various additional assumptions on the types of obstacles (see reviews in \cite{DZ17,I07,Z99,Z17}). Using the terminology of the present paper,
the works \cite{I83,G88,SZ95,SZ99,I07} give the description of one structural level 
for the intersection of $\Si (H_\OO)$ with a certain horizontal or semi-cubic strip adjacent to the real line.

 In more details, the case of two strictly convex obstacles \cite{I83,G88,I07} resembles 
 to some extend the case where a quantum graph has  
$\mu^{\min} =0$ and commensurable lengths of edges (see  Section \ref{sss:KSH}) 
because in the both cases there exists 
a horizontal strip containing an infinite number of resonances composing asymptotically horizontal sequences.
One of the differences between these two cases is that, for a quantum graph, there exists a horizontal strip 
$\{ |\im z| \le \wt \ga\}$ such that the consideration of any wider horizontal strip adds only a 
finite number of resonances, while in the case of two strictly convex obstacles \cite{G88,I07},
taking wider and wider horizontal strips one obtains more and more additional 
asymptotically horizontal sequences.

Reading the works on the case of one strictly convex obstacle $\OO$,
one could guess in \cite{SZ93,SZ95,SZ99} (see also references therein and \cite{J15})
 a program on the study  of the structure of $\Si (H_\OO)$ in regions 
adjacent to $\RR$.
One of the instruments in this program  
is the counting function in various shaped strips
\begin{gather}
\wt \N_\vphi (R) = \# \{ k \in \Si (H) \ : \ - \vphi (|\re k|) \le \im k , \ |k| \le R \} ,
\end{gather}
where the function $\vphi : [0,+\infty) \to \RR$ describes the shape of the strip 
$\{- \vphi (|\re k|) \le \im k \le  0 \}$ (see \cite{SZ95}). In particular, \cite{SZ95,SZ99} work with 
functions $\vphi $ of the form $\vphi (\xi)=\vphi_\mu (\xi) = \mu \xi^{1/3}$, $\mu \ge 0$;
let us denote the corresponding family of 
counting functions $\wt \N_\vphi (R)$ by $\N^\cub (\mu,\cdot)$.
Introducing similarly to (\ref{e:AdLog}) the asymptotic density functions for cubic semi-strips
\[
\d_{\al}^\cub (\mu) := \limsup_{R \to +\infty} \frac{\N^\cub (\mu,R)}{R^\al} , \quad \mu \in \RR , 
\]
where an additional parameter $\alpha >0$ takes into account the possible polynomial growth 
\cite{M83,Z89JFA,SZ93,SZ99},
we can infer from the results of \cite{HL94,SZ95,SZ99} that, in the case of a strictly convex obstacle $\OO \subset \RR^m$ 
satisfying additional pinching conditions of \cite{SZ99} on curvatures of the boundary of $\OO$,
the support of the measure $d \d_{m-1}^\cub (\mu)$ 
is separated from $0$ and there exists 
a partition $ 0 =\mu_0 < \mu_1 < \mu_2 < \dots < \mu_{2 n+1}  < \mu_{2n+2} = +\infty $ 
such that $\supp d \d_{m-1}^\cub (\mu) \cap [\mu_{2j-1}, \mu_{2j} ] \neq \varnothing$ and 
$\supp d \d_{m-1}^\cub (\mu) \cap [\mu_{2j}, \mu_{2j+1} ] = \varnothing$ , $j=0,\dots,n$.
This suggests that, from the point of view of the asymptotics of `narrow resonances' (see Section \ref{s:Physics}),
 the two following constants 
have to play a special role: $\mu^{\min} = \inf \,  \supp d \d_{m-1}^\cub (\mu)$ and 
$\lim_{\mu \to \mu^{\min}+0} \frac{\d_{m-1}^\cub (\mu)}{\mu - \mu^{\min}}$.

It seems that very little is known about the internal structure of $\Si (H)$ for multi-dimensional 
Schrödinger operators
$H=-\De + V $ in $L^2_\CC (\RR^m)$ in the case 
of real-valued compactly supported potentials $V \in L_{\RR,\comp}^\infty (\RR^m)$ 
and an odd number $m>1$.
Most of studies of this case were concentrated on 
the asymptotics of $\N_H (R)$ for $R \to \infty$.
On one hand, it follows from 
\cite{Z89DMJ} that $\limsup_{R\to\infty} \N_H (R)/R^m < +\infty$ for every $V \in L^\infty_{\RR,\comp} (\RR^m)$,
and from \cite{CH05} that $\limsup_{R\to\infty} \ln \N_H (R)/\ln R = m $ 
for generic $V $ in the same class. However, the best known lower bound for nontrivial smooth $V (\cdot)$
is 
$\limsup_{R\to\infty} \frac{\N_H (R)}{R} > 0$ \cite{SaB01}.
A very stimulating and intriguing discussion of the existing gap between known upper and lower bounds 
on the growth of $\N_H (R)$ can be found in \cite[Section 2.7]{Z17}. It seems that an example of 
a nontrivial $V \in L^\infty_{\RR,\comp} (\RR^m)$ with $\limsup_{R\to\infty} \ln \N_H (R)/\ln R < m $ 
is not known (see \cite[Conjecture 1 in Section 2.7]{Z17}). 
Looking from this point of view on the point interaction Hamiltonians $H_{a,Y}$, one sees that,
 loosely speaking, the lower bound of \cite{SaB01} is achieved on them in the sense that $\lim_{R\to\infty} N_{H_{a,Y}} (R)/R 
\in \RR_+$ \cite{LL17}. While $H_{a,Y}$ do not belong to the class of operators 
$H=-\De + V $ with $V \in L_{\RR,\comp}^\infty (\RR^3)$, it is reasonable to test on them any prospective 
method of proving the equality $\limsup_{R\to\infty} \ln \N_H (R)/\ln R= m $, which was conjectured in \cite[Section 2.7]{Z17}.

\vspace{1ex}
\noindent
\textbf{Acknowledgements.} 
The authors are grateful to Vladimir Lotoreichik for 
a stimulating discussion on the paper \cite{LL17}, to Olaf Post for an interesting discussion of spectral problems 
on graph-like structures, and to 
Plamen Stefanov for informing us about the program of \cite{SZ93,SZ95,SZ99,J15}. The second  named author (IK) is grateful to Herbert Koch for the hospitality of the University of Bonn, 
to Jürgen Prestin for the hospitality of the University of Lübeck, and to Dirk Langemann for the hospitality of 
TU Braunschweig.
During various parts of this research, IK was supported by the Alexander von Humboldt Foundation,
 by the VolkswagenStiftung project “Modeling, Analysis, and Approximation 
Theory toward applications in tomography and inverse problems”, and by 
the  WTZ grant 100329049 ''Mathematical Models for Bio-Medical Problems'' 
jointly sponsored by BMBF (Germany) and MES (Ukraine).


\begin{thebibliography}{99}

\small

\bibitem{AGH82} 
S. Albeverio, F. Gesztesy, R. Høegh-Krohn, 
The low energy expansion in nonrelativistic scattering theory, 
Ann. Inst. H. Poincaré Sect. A (N.S.) 37 (1982), no.1, 1--28. 


\bibitem{AGHH12} S. Albeverio, F. Gesztesy, R. Høegh-Krohn, H. Holden,   
Solvable models in quantum mechanics. 2nd edition, with an appendix by P. Exner. 
AMS Chelsea Publishing, Providence, RI, 2012.

\bibitem{AGHS83} S. Albeverio, F. Gesztesy, R. Høegh-Krohn, L. Streit,  
Charged particles with short range interactions, 
Ann. Inst. H. Poincaré Sect. A (N.S.) 38(1983), no.3, 263--293.

\bibitem{AH84} 
S. Albeverio, R. Høegh-Krohn, Perturbation of resonances in quantum mechanics,
J. Math. Anal. Appl. 101 (1984),  491--513.

\bibitem{AK17} S. Albeverio, I.M. Karabash, Resonance free regions and non-Hermitian spectral optimization 
for  Schrödinger  point  interactions, Operators and Matrices 11 (2017), no.4, 1097--1117.

\bibitem{AK18} S. Albeverio, I.M. Karabash, 
Generic asymptotics of resonance counting function for 
Schr\" odinger point interactions, to appear in “Analysis as a tool in Mathematical Physics: in Memory of B. Pavlov”, ed. Kurasov, P., Laptev, A.,
Naboko, S., Simon, B., to be published by Birkhäuser, 2020; see also the preprint, arXiv:1803.06039, 2018.


\bibitem{AK00} S. Albeverio, P. Kurasov. Singular perturbations of differential operators: solvable Schrödinger-type operators. 
Cambridge University Press, 2000.


\bibitem{BC63} R.E. Bellman, K.L. Cooke, Differential-difference equations. Academic Press. New York, London. 1963. 

\bibitem{BG12} C.A. Berenstein, R. Gay, Complex analysis and special topics in harmonic analysis. 
Springer Science \& Business Media, 2012.

\bibitem{BK13} G. Berkolaiko, P. Kuchment, 
Introduction to quantum graphs. American Mathematical Soc., 2013.


\bibitem{CH05} T.J. Christiansen, P.D. Hislop, The resonance counting function for Schrödinger operators 
with generic potentials, Math. Res. Lett. 12(6) (2005), 821--826.


\bibitem{CZ95}
S. Cox, E. Zuazua, 
The rate at which energy decays in a string
damped at one end,
Indiana Univ. Math. J. 44 (1995), no.2, 545--573.

\bibitem{DEL10} E.B. Davies, P. Exner, J. Lipovský, Non-Weyl asymptotics for quantum 
graphs with general coupling conditions, J. Phys. A  43 (2010), 474013, 16 p.

\bibitem{DP11} E.B. Davies, A. Pushnitski,  Non-Weyl resonance asymptotics for quantum graphs,
 Analysis \& PDE 4 (2011), 729--756.
 
  \bibitem{DK07} Y.N. Demkov, P.B. Kurasov, Von Neumann-Wigner theorem: Level repulsion and degenerate eigenvalues, 
Theoretical and Mathematical Physics 153(1)  (2007), 1407-1422.

\bibitem{dVT18} Y.C. de Verdi\`{e}re, F. Truc. Topological resonances on quantum graphs. 
 Ann. Henri Poincaré 19  (2018), no.5, 1419--1438. 
 

\bibitem{DZ17} S. Dyatlov, M. Zworski, Mathematical theory of scattering resonances, to be published by American Mathematical
Soc..

\bibitem{G88} C. Gérard, 
Asymptotique des pôles de la matrice de scattering pour deux obstacles strictement 
convexes. 
 Mém. Soc. Math. France (N.S.), no.31 (1988), 146 pp. 

\bibitem{E84} V. Enss, Summary of the conference and some open problems. 
In: Albeverio S., Ferreira L.S., Streit L. (eds.) ``Resonances -- Models and Phenomena''.  Springer, Berlin, Heidelberg, 1984.

\bibitem{E12} P. Exner, Open quantum systems and Feynman integrals. Springer Science \& Business Media, Berlin, 2012.

\bibitem{EL10} P. Exner, J. Lipovský, 
Resonances from perturbations of quantum graphs with rationally related edges,
J. Phys. A 43(10) 
(2010), p.105301.

\bibitem{EL17} P. Exner, J. Lipovský, Pseudo-orbit approach to trajectories of resonances in quantum graphs with general vertex coupling: Fermi rule and high-energy asymptotics, 
 J. Math. Phys. 58(4) (2017), p.042101.

\bibitem{F98} 
R. Froese, Upper bounds for the resonance counting function of Schrödinger operators in odd dimensions,
 Canad. J. Math. 50 (1998), no.3, 538--546.

\bibitem{GSS13} S. Gnutzmann, H. Schanz, U. Smilansky, 
Topological resonances in scattering on networks (graphs), Physical Review Letters 110(9) (2013), p.094101.

\bibitem{HL94} T. Harge, G. Lebeau, Diffraction par un convexe,  Invent. Math. 
118(1) (1994), 161--196.

\bibitem{H00} M. Harmer, Hermitian symplectic geometry and extension theory, 
J. Phys. A 33(50), (2000), p.9193.

\bibitem{HL17} M. Holzmann, V. Lotoreichik, 
Spectral analysis of photonic crystals made of thin rods, Asymptotic Analysis 110 (1-2) (2018), 83--112.

\bibitem{I07} 
A. Iantchenko, Scattering poles near the real axis for two strictly convex obstacles, 
 Ann. Henri Poincaré 8 (2007), no.3, 513--568.

\bibitem{I83} M. Ikawa, On the poles of the scattering matrix for two strictly convex obstacles,
 J. Math. Kyoto Univ. 23(1) (1983), 127--194.

\bibitem{J15} Long Jin, Scattering resonances of convex obstacles for general boundary conditions,
Comm. Math. Phys. 335(2) (2015), 759--807.

\bibitem{I86} A. Intissar, A polynomial bound on the number of scattering poles for a potential in even
dimensional space $R^n$, Comm. Partial Differential Equations 11 (1986), no.4, 367--396.

\bibitem{K13} I.M. Karabash,
Optimization of quasi-normal eigenvalues for 1-D wave equations in inhomogeneous media;
description of optimal structures,  Asymptotic Analysis 81 (2013) no.3-4, 273-295.
 


\bibitem{Ka14} I.M. Karabash, Pareto optimal structures producing resonances of minimal decay under $L^1$-type constraints,
J. Differential Equations
257 
(2014), 
374--414.



\bibitem{KLV17} I.M. Karabash, O.M. Logachova, I.V. Verbytskyi, Nonlinear bang-bang eigenproblems and optimization of resonances in layered cavities, 
Integr. Equ. Oper. Theory 88(1) (2017), 15--44.


\bibitem{KS99} V. Kostrykin, R. Schrader, Kirchhoff's rule for quantum wires, 
J. Phys. A 32(4) (1999), p.595.


\bibitem{K68} 
A.O. Kravitsky, The two-fold expansion of a certain non-selfadjoint boundary value problem in series 
of eigenfunctions, 
Differentsial'nye Uravneniya 4(1) (1968), 165--177 (Russian).



\bibitem{Lang02} S. Lang, Algebra.  Springer-Verlag, New York, 2002.




\bibitem{LZ16} Minjae Lee, M. Zworski,  A Fermi golden rule for quantum graphs,
 J. Math. Phys. 57(9), (2016),  p.092101.

\bibitem{L16} 
J. Lipovský, Quantum Graphs And Their Resonance Properties, Acta Physica Slovaca 66(4) (2016), 265--363.

\bibitem{LL17} 
J. Lipovský,  V. Lotoreichik, 
Asymptotics of Resonances Induced by Point Interactions, Acta Physica Polonica A 132 
(2017), 1677--1682.

\bibitem{M83} R.B. Melrose, 
Polynomial bound on the number of scattering poles,  J. Funct. Anal. 53 (1983), no.3,
287--303.

\bibitem{P97} V.N. Pivovarchik, 
Inverse problem for a smooth string with damping at one end,
J. Operator Theory 
38
(1997), no.2, 243--263.


\bibitem{P12}
O. Post, Spectral analysis on graph-like spaces. Springer Science \& Business Media, 2012.

\bibitem{RSIV78} M. Reed, B. Simon, Analysis of Operators, Vol. IV of Methods of Modern Mathematical Physics.  Academic Press, New York, 1978. 


\bibitem{SaB01} A. Sá Barreto, Remarks on the distribution of resonances in odd dimensional Euclidean scattering. 
Asymptotic Analysis 27(2), (2001), 161--170.


\bibitem{S96} M.A. Shubov, 
Basis property of eigenfunctions of nonselfadjoint operator pencils generated by the equation of 
nonhomogeneous damped string, Integr. Equ. Oper. Theory 25(3)  (1996), 289--328.

\bibitem{SZ93} J. Sjöstrand, M. Zworski, Estimates on the number of 
scattering poles near the real axis for strictly convex obstacles, Ann. Inst. Fourier 43(3), 
(1993), 769--790.

\bibitem{SZ95} J. Sjöstrand, M. Zworski,  
The complex scaling method for scattering by strictly convex obstacles, 
Arkiv för Matematik 33(1), (1995), 
135--172.

\bibitem{SZ99} 
J. Sjöstrand, M. Zworski,  
Asymptotic distribution of resonances for convex obstacles, Acta Mathematica 183(2), 
(1999), 191--253.


\bibitem{S06} P. Stefanov, Sharp upper bounds on the number of the scattering poles,
 J. Funct. Anal. 231(1) (2006),
111--142.

\bibitem{V94_DMJ} G. Vodev, Sharp bounds on the number of scattering poles in even-dimensional spaces, 
Duke Math. J. 74(1)
(1994), 1--17.

\bibitem{Z87} M. Zworski, 
Distribution of poles for scattering on the real line, 
 J. Funct. Anal. 73(2) (1987), 277--296.
 
 \bibitem{Z89JFA} M. Zworski, Sharp polynomial bounds on the number of scattering poles of radial potentials,
   J. Funct. Anal.  82 (1989), no.2, 370-403.
  
\bibitem{Z89DMJ} M. Zworski, Sharp polynomial bounds on the number of scattering 
poles, Duke Math. J. 59(2),  (1989),  311--323.
  
\bibitem{Z99}  M. Zworski,  Resonances in physics and geometry, Notices of the AMS 46(3) (1999), 319--328.

\bibitem{Z17} M. Zworski, Mathematical study of scattering resonances, 
 Bull. Math. Sci. 7(1)  (2017), 1--85.

\end{thebibliography}
\end{document}